\documentclass[12pt]{article} 

\usepackage{amsmath,amssymb,amsfonts,amsthm,mathrsfs}

\usepackage{kpfonts}

\usepackage[margin=1in]{geometry}
\usepackage{natbib}
\setcitestyle{authoryear,open={(},close={)}}
\usepackage[breaklinks=true]{hyperref}
\usepackage{breakcites}
\nopagebreak

\usepackage{setspace}
\usepackage{color}

\usepackage{hyperref}
\hypersetup{
  colorlinks=true,
  linkcolor=blue,
  citecolor=red
}

\usepackage{array}
\usepackage{tabularx}
\newcolumntype{L}[1]{>{\raggedright\arraybackslash}p{#1}}
\newcolumntype{C}[1]{>{\centering\arraybackslash}p{#1}}
\newcolumntype{R}[1]{>{\raggedleft\arraybackslash}p{#1}}

{\theoremstyle{definition}}
\newtheorem{theorem}{Theorem}
\newtheorem{lemma}{Lemma}
\newtheorem{claim}{Claim}

\newtheorem{definition}{Definition}
\newtheorem{proposition}{Proposition}
\newtheorem{corollary}{Corollary}

\makeatletter
\newenvironment{proofof}[1]{\par
	\pushQED{\qed}%
	\normalfont \topsep6\p@\@plus6\p@\relax
	\trivlist
	\item[\hskip\labelsep
	\bfseries
	Proof of #1\@addpunct{.}]\ignorespaces}
	{\popQED\endtrivlist\@endpefalse}
\makeatother

\usepackage{tikz}
\usetikzlibrary{matrix}
\usetikzlibrary{decorations.pathreplacing,angles,quotes}
\tikzset{
	solid node/.style={circle,draw,inner sep=1.0,fill=black},
	hollow node/.style={circle,draw,inner sep=1.0}}
\tikzset{
	invisible/.style={opacity=0},
	visible on/.style={alt=#1{}{invisible}},
	alt/.code args={<#1>#2#3}{
		\alt<#1>{\pgfkeysalso{#2}}{\pgfkeysalso{#3}}},}
\usetikzlibrary{arrows}
\usetikzlibrary{shapes}

\usepackage{pgfplots} 
\pgfplotsset{compat=newest}
\pgfplotsset{plot coordinates/math parser=false}
\pgfplotsset{
	every non boxed x axis/.style={
	xtick align=center,
	enlarge x limits=true,
	x axis line style={line width=0.8pt, -latex}
},
    every boxed x axis/.style={}, enlargelimits=false
}
\pgfplotsset{
	every non boxed y axis/.style={
	ytick align=center,
	enlarge y limits=true,
	y axis line style={line width=0.8pt, -latex}
},
	every boxed y axis/.style={}, enlargelimits=false
}
\usetikzlibrary{
	arrows.meta,
	intersections}

\def\N{\mathbb{N}}

\def\cS{\mathcal{S}}

\title{Repeated Coordination with Private Learning%
  \thanks{We thank Drew Fudenberg, George Mailath, Manuel
    Mueller-Frank, Roger Myerson, Philipp Strack, and Adam Wierman for
    comments. This work was supported by a grant from the Simons
    Foundation (\#419427, Omer Tamuz).}} \author{Pathikrit Basu%
  \thanks{California Institute of Technology. E-mail -
    pathkrtb@caltech.edu} \space Kalyan Chatterjee%
  \thanks{Pennsylvania State University. E-mail - kchatterjee@psu.edu}
  \space Tetsuya Hoshino%
  \thanks{Pennsylvania State University. E-mail - tzh144@psu.edu}
  \space Omer Tamuz%
  \thanks{California Institute of Technology. E-mail -
    omertamuz@gmail.com} } \date{\today}

\begin{document}
\maketitle

\begin{abstract}
  We study a repeated game with payoff externalities and observable
  actions where two players receive information over time about an
  underlying payoff-relevant state, and strategically coordinate their
  actions. Players learn about the true state from private signals, as
  well as the actions of others. They commonly learn the true
  state~\citep{cems1}, but do not coordinate in every equilibrium.  We
  show that there exist stable equilibria in which players can overcome 
  unfavorable signal realizations and 
  eventually coordinate on the correct action, for any discount
  factor. For high discount factors, we show that in addition players
  can also achieve efficient payoffs.
\end{abstract}

\section{Introduction}

We consider a repeated coordination game with observable actions and
with unobserved stage utilities that depend on an unknown, unchanging
state of nature. The state is either high ($H$) or low ($L$). In each
period, players choose between two actions: invest ($I$) or not invest
($N$). Not investing is a safe action that always yields stage utility
0, whereas investment involves a cost. Investment yields a profit only
if the state is high and the other player also invests.

In each period, each player observes a private signal regarding the
state, and chooses whether to invest or not. These signals are
independent, and are observed regardless of the actions taken by the
players.\footnote{Equivalently, when players both invest they receive independent random payoffs that depend on the state, but even when they do not invest, they observe the (random) payoff they would have gotten from investing. One could alternatively consider a setting in which players receive signals from distributions that also depend on their actions. We believe that our results hold as long as these signals are informative for any action.}
Thus, for each player the uncertainty regarding the state
eventually vanishes, through the aggregation of these
signals. However, because signals are private, a player does not know
what information the other player has, and hence does not know exactly
what his beliefs are. The question that now arises is: can the players
coordinate on the right action? That is, can they eventually
coordinate on investment when the state is high, and on not investing
when the state is low?  While common learning~\citep{cems1}---whose
introduction was motivated by games of this nature---is obtained in
this game, it is far from obvious that coordination is possible. 
In particular, it is not obvious that there exist {\em stable} equilibria, 
in which, when the state is high, players can recover from events of 
low signal realizations that cause coordination to stop  \citep[See, e.g., ][]{ellison1994cooperation}.


In the one-shot version of our game, when signals are strong enough,
players can coordinate on the correct action, with high probability;
this follows from known results on common-$q$ beliefs in static
Bayesian games \citep{monderersamet}. In our repeated setting, since
signals accumulate over time, private information does indeed become
strong. However, agents are forward looking, and so these results do not apply directly. Moreover, a player's
beliefs are influenced both by his own private signals, as well as the
actions of the other player. Thus reasoning directly about players'
beliefs is impossible, and one needs to revert to more abstract
arguments, in the flavor of the social learning literature, in order
to analyze the dynamics. Indeed, the question that we ask regarding
eventual coordination on the correct action closely corresponds to key
questions regarding eventual learning addressed in the social learning
literature~\citep[e.g.,][]{smith2000pathological}.

We show that it is indeed possible to coordinate with high
probability. In fact, we construct equilibria with the stronger
property of {\em action efficiency}, in which, from some time period
on, the players only choose the correct action. It thus follows that these equilibria are {\em stable}: 
when the state is high, agents recover from events of low signal realizations, and eventually coordinate.
We further show that
these equilibria achieve near optimal {\em payoff efficiency} when the
discount factor is high.

Static coordination games of this kind have become an insightful way
to model many real-world problems such as currency attacks
\citep{morrisshin}, bank runs \citep{goldstein}, political regime
changes \citep{edmond}, and arms races \citep{baliga}. \cite{yang}
considers investment in a fluctuating risky asset, where the payoff
depends, as in our paper, on both the fundamentals (the state of
nature) and the action of the other player.  A common assumption in
this literature is that direct communication is not allowed. We also
make this assumption, and indeed our question is trivial without it.

We describe both equilibrium strategies and payoffs. The construction
of the equilibrium is not fully explicit; at some point we have to use
a fixed point theorem to show existence of the equilibria with the
properties we desire. 

The construction is roughly as follows: Initially, players play action
$N$ for a sufficiently long amount of time, enough to accumulate a
strong private signal. After this, an {\em investment phase} begins, in
which players invest if their beliefs are above some threshold
level. In the high state, both players are very likely to be above
this threshold, while in the low state they are both likely to be
below it.

The investment phase continues until one (or both) of the
players choose $N$, in which case a {\em cool-off phase} of predetermined
length is triggered. During the cool-off phase the players do not
invest, and instead again accumulate signals, so that, at the end of
the the cool-off phase and the beginning of the next investment phase,
they are again likely to choose the correct action.

In the high state, we show that each investment phase can, with high
probability, last forever. Indeed, we show that almost surely,
eventually one of these investment phases will last forever, and so
the players eventually coordinate on the correct action.  In the low
state, we show that eventually players will stop investing, as their
beliefs become too low to ever pass the investment thresholds. Thus,
almost surely, our equilibria display {\em action efficiency}, with
the players coordinating on the right action.

Action efficiency is achievable for any discount factor. However,
because of the initial phase in which players do not invest, there is
a significant loss in payoff for low discount factors. For high
discount factors, we show that there is almost no loss in payoff, as
compared to agents who start the game knowing the state, and
coordinate efficiently.

To prove the existence of our equilibria, we first restrict the
players to strategies that implement the above described cool-off and
investment phases, and apply a standard fixed point theorem to find an
equilibrium in the restricted game. We then show that equilibria of
the restricted  game are also equilibria of the unrestricted game. To apply this proof strategy we frame our discussion in terms of  mixed strategies rather than the more commonly used behavioral strategies.

We shall now consider briefly some relevant literature. In Section \ref{section: model}, we formally describe the model. Section \ref{section: equilibrium analysis} contains the analysis and construction of the equilibrium. Section \ref{section: discussion} discusses specific aspects of the model in greater detail and Section \ref{section: conclusion} concludes. Proofs are relegated to the Appendix unless otherwise mentioned.

\subsection*{Related literature}
There is work related to this paper in three different parts of the
economics literature.

A very similar setting of a repeated coordination game with private
signals regarding a state of nature is offered as a motivating example
by \cite{cems1}, who develop the notion of \textit{common
  learning} to tackle this class of repeated games. We believe that
the techniques we develop to solve our game should be easily
applicable to theirs.

The literature on \textit{static coordination with private
  information} mentioned earlier in the introduction has dynamic
extensions. The dynamic game of regime change by \cite{angeletos}, for
example, has a similar set-up in terms of private signals and publicly
observable (in aggregate) actions, in a repeated game with a continuum
of agents. In their setting, the game ends once the agents coordinate
on a regime change. In contrast, in our setting one period of
successful investment does not end the game. \cite{chassang} is also
related, but differs from our setting in that is has a state of nature
that changes over time, and an action that is irreversible.

The recent literature on \textit{folk theorems} in repeated settings
is possibly the closest, in its statement of the problem addressed, to
our paper. The papers we consider in some detail below are
\cite{yamamoto} and \cite{sugayayamamoto}.  Most of the other existing
work on learning in repeated games assumes that players observe public
signals about the state, and focuses on equilibrium strategies in
which players ignore private signals and learn only from public
signals \citep{wiseman2005partial, wiseman2012partial,
  fudenberg2010repeated,fudenberg2011learning}.

\cite{yamamoto} and \cite{sugayayamamoto} are the closest papers to
ours. However, these papers are significantly different along the
following dimensions. First, they focus on ex-post equilibria, which
are defined as sequential equilibria in the infinitely repeated game
in which the state $\theta$ is common knowledge for each
$\theta$. Since an ex-post equilibrium is a stronger notion, their
folk theorem (based on ex-post equilibria) implies the folk theorem
based on sequential equilibria. Their results do not apply to our
setting, since they make assumptions that our model does not
satisfy. Second, for their folk theorem, they assume that the set of
feasible and individually rational payoffs has the dimension equal to
the number of players for \emph{every} state. However, this assumption
fails to hold in our environment. Third, their folk theorem fails to
hold when there are two players and private signals are conditionally
independent. Since our work considers such cases, it is complementary
to theirs also in this sense.

\section{Model}
\label{section: model}
We consider an infinitely repeated game of incomplete
information. There is an unknown payoff-relevant binary state of
nature $\theta \in \Theta = \{H,L\}$. Time is discrete and the horizon
is infinite, i.e., $t = 1,2,3 ...$. There are two players
$\mathcal{N} = \{1,2\}$. Each player $i$, at every period $t$, chooses
an action $a_i \in A_i = \{I,N\}$. Action $I$ can be interpreted as an
action to \textit{invest} and action $N$ as \textit{not invest}.  We
assume that the state of nature is determined at period $t = 1$ and
stays fixed for the rest of the game. Further, both players hold a
common prior $(p_0,1-p_0) \in \Delta(\Theta)$, where $p_{0} \in (0,1)$
is the initial probability that the state is $\theta = H$.  The payoff
to player $i$, in every period, from any action profile
$a \in A = A_1 \times A_2$ depends on the state and  is defined by a
per-period payoff function $u_i(a,\theta)$. The following table shows
the payoffs from actions associated with states $\theta=H$ and
$\theta = L$ respectively: \renewcommand{\arraystretch}{1.1}
\begin{table}[h]
\centering
\begin{minipage}[t]{0.45\textwidth}
	\centering
	\begin{tabular}{C{10mm}C{20mm}C{15mm}}
			\hline\hline
			\multicolumn{1}{r}{}
		&	\multicolumn{1}{c}{$I$}
		&	\multicolumn{1}{c}{$N$}
		\\	\hline
			$I$
		&	$1-c,1-c$
		&	$-c,0$
		\\	$N$
		&	$0,-c$
		&	$0,0$
	\end{tabular}
	\caption{$\theta = H$}
\end{minipage}
\begin{minipage}[t]{0.45\textwidth}
	\centering
	\begin{tabular}{C{10mm}C{20mm}C{15mm}}
			\hline\hline
			\multicolumn{1}{r}{}
		&	\multicolumn{1}{c}{$I$}
		&	\multicolumn{1}{c}{$N$}
		\\	\hline
			$I$
		&	$-c,-c$
		&	$-c,0$
		\\	$N$
		&	$0,-c$
		&	$0,0$
	\end{tabular}
	\caption{$\theta = L$}
\end{minipage}
\end{table}
\renewcommand{\arraystretch}{1.0}

\noindent
Here $c \in (0,1)$ is a cost of investing, and so investment is
beneficial if and only if the other player chooses to invest and the
state is $H$. In this case both players obtain a payoff equal to
$1-c$. In the state $H$, there are two stage-game equilibria in pure
strategies, $(I,I)$ and $(N,N)$, where the former is Pareto
dominant. In the state $L$, however, action $N$ is strictly dominant,
thereby the unique stage-game equilibrium being $(N,N)$. We assume
that stage utilities are not observed by the agents.

\paragraph{Private Signals.}
At every period $t$, each player $i$ receives a \textit{private signal} $x_{it} \in X$ about the state, where $X$ is a finite signal space.
\footnote{The assumption that the signal space is finite is not crucial for establishing our results. The key feature of the signal distributions that we exploit is that they generate bounded likelihood ratios.}
The signal distribution under the state $\theta$ is given by $f_{\theta} \in \Delta(X)$.
We assume that conditional on the state $\theta$, private signals are \textit{independent} across players and across time. We also assume that $f_H \neq f_L$. This allows both players to individually learn the state of the nature asymptotically via their own private signals. Finally, we assume full support: $f_\theta(x) > 0$ for all $\theta \in \Theta$ and for all $x \in X$. Hence, no private signal fully reveals the state of nature.

\paragraph{Histories and Strategies.}
At every period $t$, each player $i$ receives a private signals $x_{it}$ and then chooses action $a_{it}$. 
Under our assumption of perfect monitoring of actions, player $i$'s \textit{private history} before taking an action $a_{it}$ at period $t$ is hence $h_{it} = (a^{t-1},x_{i}^{t})$, where we denote by $a^{t-1} = (a_{1},a_{2},\ldots,a_{t-1}) \in A^{t-1}$ the past action profiles and by $x_{i}^{t} = (x_{i1},x_{i2},\ldots,x_{it}) \in X^{t}$ player $i$'s private signals. Let $H_{it} = A^{t-1} \times X^t$ be the set of all period-$t$ private histories for player $i$. Note that action history $a^{t-1}$ is publicly known. Figure \ref{fig: timeline} depicts the sequence of events. A \textit{pure strategy} for player $i$ is a function $s_{i} : \bigcup_{t} H_{it} \rightarrow A_i$. Let $S_{i}$ denote the set of all player $i$'s pure strategies. We equip this set with the product topology. Let $\Sigma_i = \Delta(S_i)$ denote the set of mixed strategies, which are the Borel probability measures on $S_i$. 

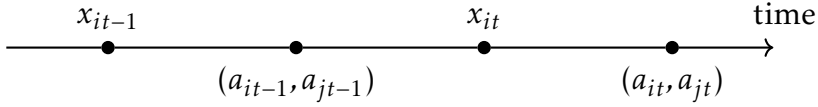
\begin{figure}
\centering
\begin{tikzpicture}[]
\node[]{}
child[level distance=15mm, grow=right, thick]{
    node[solid node, inner sep=1.5pt, label=above:{$x_{it-1}$}]{}
    child[level distance=25mm, grow=right, thick]{
        node[solid node, inner sep=1.5pt, label=below:{$(a_{it-1},a_{jt-1})$}]{}
        child[level distance=25mm, grow=right, thick]{
            node[solid node, inner sep=1.5pt, label=above:{$x_{it}$}]{}
            child[level distance=25mm, grow=right, thick]{
                node[solid node, inner sep=1.5pt, label=below:{$(a_{it},a_{jt})$}]{}
                child[level distance=15mm, grow=right, thick, ->]{
                    node[label=above:{time}]{}
                }
            }
        }
    }
};
\end{tikzpicture}
\caption{Timeline}
\label{fig: timeline}
\end{figure}

\paragraph{Payoffs and equilibria.} 
Both players evaluate infinite streams of payoffs via discounted sum where the discount factor is $\delta \in (0,1)$. Hence, in the infinitely repeated game, if the players choose a path of actions $(a_t)_t \in A^{\infty}$, then player $i$ receives long-run payoff $\sum_{t \geq 1}(1-\delta)\delta^{t-1} u_i(a_t,\theta)$ when the state is $\theta$. Hence, under a (mixed) strategy profile $\sigma$, the long-run expected payoff to the player is 
\begin{align}
\mathbb{E}_{\sigma}\Big[\displaystyle\sum_{t \geq 1}(1-\delta)\delta^{t-1} u_i(a_t,\theta)\Big]
\end{align}
where the expectation operator, $\mathbb{E}_{\sigma}$, is defined with respect to the probability measure $\mathbb{P}_{\sigma}$ induced by the strategy profile $\sigma$ on all possible paths of the repeated interaction $\Omega = \Theta \times (X^2)^{\infty}\times A^{\infty}$. As usual, a strategy profile is a Nash equilibrium if no player has a profitable deviation. In our setting, we show that for every Nash equilibrium there exists an outcome-equivalent sequential equilibrium (Lemma~\ref{lemma:sequential}).

\section{Equilibrium Analysis}
\label{section: equilibrium analysis}
In this section, we study equilibria of the infinitely repeated
game. Our goal is to establish the existence of equilibria that are
efficient. We shall consider two notions of efficiency, namely
\textit{action} efficiency and \textit{payoff} efficiency. Action
efficiency requires that player actions converge almost surely to the
correct actions, i.e., both players eventually invest (resp.\ not
invest) under state $\theta = H$ (resp.\ $\theta = L$). Payoff
efficiency requires that players obtain payoffs close to efficient
payoffs, i.e., in state $\theta = H$ (resp.\ $\theta = L$), ex-ante
long run payoffs of both players are close to $1-c$ (resp.\ close to
$0$). The main result is that action efficiency can always be obtained
in equilibrium and payoff efficient equilibria exist for discount
factor $\delta$ close to one.

\subsection*{One-period Example}
To informally illustrate the nature of incentives that arise in our
coordination framework, we first consider the scenario when there is
only a single period of interaction. Suppose that players follow a
threshold rule to coordinate on actions. That is, for each player $i$,
there exists a threshold belief $\bar{p}_i$ such that player $i$
invests (does not invest) if his belief about the state $\theta = H$
is above (below) the threshold. For simplicity, let us further assume
that no posterior belief after the one period signal equals the
threshold. Now, if player $i$'s belief $p_i$ is above his threshold
and he chooses to invest, then he obtains a payoff equal to
\begin{align*}
    p_i\mathbb{P}[p_j > \bar{p}_j \mid \theta = H] - c,
\end{align*}
where $\mathbb{P}[p_j > \bar{p}_j \mid \theta = H]$ is the probability
that player $j$ invests conditional on the state $\theta = H$. Since
player $i$ obtains payoff of $0$ from not investing, he has an
incentive to invest only if this probability is high enough. Moreover,
even if it were the case that the other player invests for sure in the
state $H$, player $i$ would still need his belief $p_i$ to be above
$c$ in order to find investment optimal. Hence, when player $j$
invests with high probability in the state $H$, the threshold of
player $i$ would be close to $c$. Likewise, if player $j$ invests with
low probability in the state $H$, the threshold would be close to
one. There may hence arise multiple equilibria. There could be an
equilibrium with low thresholds where both players invest with high
probability in the state $H$. There could also be an equilibrium with
high thresholds where this probability is low or even zero.
\footnote{This multiplicity arises from the assumption that investment
  cannot be strictly dominant in either state, unlike global games
 \citep{carlsson1993global}.}

Note the roles played by both a player's belief about the state, and
by his belief about his opponent's beliefs. If player $i$ believes
that the state is $\theta = H$ with probability $1$, his payoff from
investing is $\mathbb{P}[p_j > \bar{p}_j \mid \theta = H] - c$.  Thus,
he invests only if he believes that player $j$ invests with
probability at least $c$.  On the other hand, if he believes that
player $j$ chooses to invest with probability $1$, his payoff from
investing is $p_i - c$.  Thus, he invests only if he believes that the
state is $\theta = H$ with probability at least $c$. This demonstrates
that not only do players need high enough beliefs about the state $H$
in order to invest, but also need to believe with high probability
that their opponents hold high beliefs as well. Both players invest
only if they acquire ``approximate common knowledge'' of the state
being $\theta = H$ \citep{monderersamet}. This naturally raises the
question of whether an analogue of this insight obtains in the
repeated game we consider here.

Now, imagine that the first period signals were informationally very
precise. This would allow a very high probability of investment in the
state $H$ in the low threshold equilibrium. This suggests that if such
precise signals were the result of accumulating a series of less
precise ones, say received over time, the probability of coordinating
on the right actions would be very high. In the infinitely repeated
game, such a situation can arise, for example, when players choose to
not invest and individually accumulate their private signals, followed
by an attempt to coordinate with a low threshold.

In the dynamic context, because players are forward looking, such a
threshold rule would also depend on how players behave in the
future. For example, if non-investment triggers both players to not
invest forever, the thresholds for investment would be low. This is
because non-investment would yield a payoff of $0$ whereas investment
would yield expected payoff for the current period and a non-negative
future continuation payoff (because a player can always guarantee
himself a payoff of $0$ by never investing on his own).\footnote{In
  the model of \cite{angeletos}, once sufficiently many players invest
  (or attack a status quo), they receive payoffs and the game
  ends. Hence, the issue of ``repeated'' coordination studied here
  does not arise in their work.} However, having non-investment
trigger players not investing forever is not compatible with action
efficiency, which requires the players to invest eventually in the
high state. Thus action efficient equilibria must somehow incentivize
players to revert back to investing.

\subsection{Cool-off Strategies}
We will focus on a particular class of strategies which we call \textit{cool-off strategies}, in which play involves two kinds of phases: \emph{cool-off phases} and \emph{investment phases}. In a cool-off phase, both players accumulate private signals while playing action profile $(N,N)$ for a certain period of time. After the cool-off phase, they enter an investment phase in which each player chooses action $I$ with positive probability. In an investment phase, if the players' play action profile $(I,I)$ then they remain in the investment phase, otherwise they restart a new cool-off phase.

Cool-off strategies are based on the notion of a \textit{cool-off scheme}, which is defined as a partition $(\mathcal{C},\mathcal{I})$ of the set of all action histories: the set of cool-off histories $\mathcal{C}$ and the set of investment histories $\mathcal{I}$. This classification is carried out by means of a \textit{cool-off function} $T: \mathbb{N} \rightarrow \mathbb{N}$ which labels each action history $a^{t}$ as being part of a cool-off phase or not. Formally, this is done as follows.

\begin{definition}\label{def:cool-off scheme}
A cool-off scheme $(\mathcal{C},\mathcal{I})$ induced by a cool-off function $T : \mathbb{N} \rightarrow \mathbb{N}$ is recursively defined as follows:
\begin{enumerate}
\item The empty action history is in $\mathcal{C}$.
\item For an action history $a^t$ such that $t \geq 1$:
\begin{enumerate}
\item Suppose that $a^{t-1} \in \mathcal{I}$. Then, $a^{t} \in \mathcal{I}$ if $a_t = (I,I)$ and $a^t \in \mathcal{C}$ otherwise.
\item Suppose that $a^{t-1} \in \mathcal{C}$ and that there exists a subhistory of $a^{t}$ in $\mathcal{I}$.
Let $a^{s-1}$ be the longest such subhistory. Then, $a^t \in \mathcal{C}$ if $t \le s + T(s) - 1$ and $a^t \in \mathcal{I}$ if $t > s + T(s) - 1$.
\item Suppose that $a^{t-1} \in \mathcal{C}$ and that there does not exist a subhistory of $a^{t}$ in $\mathcal{I}$. Then, $a^t \in \mathcal{C}$ if $t \le T(1)$ and $a^t \in \mathcal{I}$ if $t > T(1)$.
\end{enumerate}
\end{enumerate}
Each series of action histories in the sets $\mathcal{C}$ and $\mathcal{I}$ is called a \emph{cool-off phase} and an \emph{investment phase} respectively.
\end{definition}

Figure \ref{fig:cool-off} depicts a cool-off scheme. Initially, players are in a cool-off phase for $T(1)$ periods. At period $T(1)+1$, they are in the investment phase. If both players choose to invest, then they stay in the investment phase. If at some point during the investment phase, say at period $s$, at least one player chooses to not invest, then another cool-off phase begins at this action history $a^{s}$ for $T(s)$ periods. Another investment phase starts at period $s+T(s)$.
\footnote{Players may never invest in an investment phase.
For example, if action $a_{T(1)+1} \neq (I,I)$ then this investment phase ends immediately, and a new cool-off phase starts at period $T(1)+2$.}

As is clear from Definition \ref{def:cool-off scheme}, the cool-off scheme $(\mathcal{C},\mathcal{I})$ induced by any cool-off function $T$ satisfies both $\mathcal{C} \cap \mathcal{I} = \varnothing$ and $\mathcal{C} \cup \mathcal{I} = \bigcup_{t \ge 1} A^{t}$. Further, whether an action history $a^{t}$ is in a cool-off phase or in an investment phase depends only on the cool-off function $T$ but not on strategies.

\begin{figure}[t]
\centering
\begin{tikzpicture}[->,>=stealth',node distance=22mm,semithick]
\tikzstyle{every node}=[draw=none,inner sep=0mm,ellipse,minimum height=10mm,font=\footnotesize]

\node[ellipse,fill=blue!15,minimum width=17.5mm] (C11) [] {$a^{1}$};
\node[ellipse,fill=blue!15,minimum width=17.5mm] (C12) [right of=C11] {$a^{2}$};
\node[minimum size=6mm] (C1d) [right of=C12] {{\normalsize ...}};
\node[ellipse,fill=blue!15,minimum width=17.5mm] (C13) [right of=C1d] {$a^{T(1)}$};
\node[ellipse,fill=red!15,minimum width=17.5mm] (I11) [below right of=C12] {$a^{T(1)+1}$};
\node[ellipse,fill=red!15,minimum width=17.5mm] (I12) [right of=I11] {$a^{T(1)+2}$};
\node[minimum size=6mm] (I1d) [right of=I12] {{\normalsize ...}};
\node[ellipse,fill=red!15,minimum width=17.5mm] (I13) [right of=I1d] {$a^{s-1}$};
\node[ellipse,fill=blue!15,minimum width=17.5mm] (C21) [below right of=I12] {$a^{s}$};
\node[ellipse,fill=blue!15,minimum width=17.5mm] (C22) [right of=C21] {$a^{s+1}$};
\node[minimum size=6mm] (C2d) [right of=C22] {{\normalsize ...}};
\node[ellipse,fill=blue!15,minimum width=17.5mm] (C23) [right of=C2d] {$a^{s+T(s)-1}$};
\node[ellipse,fill=red!15,minimum width=17.5mm] (I21) [below right of=C22] {$a^{s+T(s)}$};
\node[minimum size=6mm] (I2d) [right of=I21] {{\normalsize ...}};

\draw [->] (C11) -- (C12);
\draw [->] (C12) -- (C1d);
\draw [->] (C1d) -- (C13);
\draw [->] (C13) -- (I11);
\draw [->] (I11) -- (I12);
\draw [->] (I12) -- (I1d);
\draw [->] (I1d) -- (I13);
\draw [->] (I13) -- node [left] {$a_{s} \neq (I,I)$} (C21);
\draw [->] (C21) -- (C22);
\draw [->] (C22) -- (C2d);
\draw [->] (C2d) -- (C23);
\draw [->] (C23) -- (I21);
\draw [->] (I21) -- (I2d);
\end{tikzpicture}
\caption{Cool-off Schemes}
\label{fig:cool-off}
\end{figure}
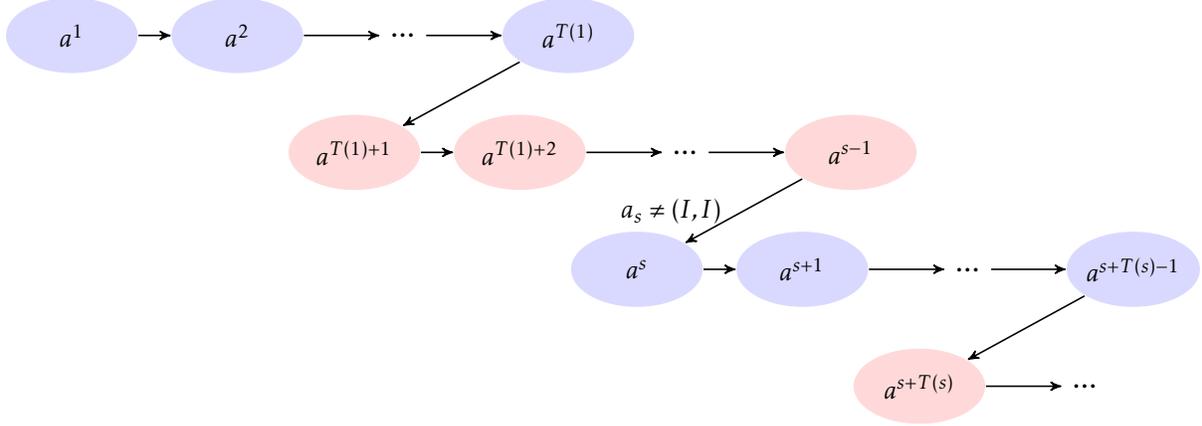

\begin{definition}\label{lem:cool-off scheme}
Fix a cool-off function $T$. We define the set $S_{i}(T)$ as the set of all player $i$'s pure strategies $s_{i} \in S_{i}$ with the following properties:
\begin{enumerate}
\item During a cool-off phase player $i$ chooses action $N$ for any realization of their private signals. That is, if $a^{t-1} \in \mathcal{C}$, then $s_{i}(a^{t-1},x_{i}^{t}) = N$ for any $x_{i}^{t} \in X^{t}$.
\item In any action history $a^{t-1}$ in which any player chose $I$ during a cool-off phase at some time $\tau< t$, player $i$ chooses $N$, i.e., $s_i(a^{t-1},x_i^t)= N$.
\end{enumerate}
Let $\Sigma_i(T)$ be the set of mixed strategies whose support is contained in $\Sigma_i(T)$.
\end{definition}

The following claim is immediate.
\begin{claim}\label{clm:cool-offs-closed}
Fix a cool-off function $T$. Then, the set $\Sigma_i(T)$ is a closed and convex subset of the set $\Sigma_i$.
\end{claim}

As stated above, the investment phase corresponds to action histories in which players would invest with positive probability. We have not required that this be the case in Definition \ref{lem:cool-off scheme}. For example, the trivial equilibrium, in which both players never invest, belongs to the class $\Sigma_{i}(T)$. Ideally, we would like strategies in the set $\Sigma_1(T) \times \Sigma_2(T)$ to be such that \textit{whenever} the players are in an investment phase, there is a large probability of investment. Hence, we focus on a smaller class of strategies, which we introduce after the following definitions.

A private history $(a^{t-1},x_i^t)$ is said to be {\em consistent} with a pure strategy $s_i$ if for all $\tau  < t$ it holds that $s(a^{\tau-1},x_i^\tau)=a_i^\tau$.  A private history $(a^{t-1},x_i^t)$ is said to be consistent with a mixed strategy $\sigma_i$ if it consistent with some pure strategy $s_i$ in the support of $\sigma_i$. Here, by support we mean the intersection of all closed subsets of $S_i$ which have probability one under $\sigma_i$; this is equivalent to requiring that for each time $t$ the restriction of $s_i$ to time periods $1$ to $t$ is chosen with positive probability. Given a private history $(a^{t-1},x_i^t)$ that is consistent with $\sigma_i$, we can define the (in a sense) conditional probability of taking action $I$, which we denote by $\sigma_i(a^{t-1},x_i^t)(I)$, as
\begin{align*}
    \sigma_i(a^{t-1},x_i^t)(I) = \frac{\sigma_i(s_i \text{ is consistent with } (a^{t-1},x_i^t) \text{ and } s_i(a^{t-1},x_i^t) = I)}{\sigma_i(s_i \text{ is consistent with } (a^{t-1},x_i^t))}\cdot
\end{align*}
Note that this is simply the probability of choosing $I$ at private history $(a^{-1},x_i^t)$ under the behavioral strategy that corresponds to $\sigma_i$ \citep[see, e.g., ][]{osborne1994course}. We say that a action history $a^{t-1}$ is consistent with $\sigma_i$ if $(a^{t-1},x_i^t)$ is consistent with $\sigma_i$ for some $x_i^t$ that occurs with positive probability. Since private signals are conditionally independent, we can define the conditional probability of investing at a consistent $a^{t-1}$ in the high state as
\begin{align*}
    \sigma_i(a^{t-1})(I|H) = \frac{\displaystyle\sum_{x^t_i \in X^t} f^H(x^t_i) \Big[\prod_{\tau=1}^{t-1} \sigma_i(x_i^\tau,a^{\tau-1})(a_{i\tau})\big(\sigma_i(x^t_i,a^{t-1})(I)\big)\Big]}{\displaystyle\sum_{x^t_i \in X^t} f^H(x^t_i) \Big[\prod_{\tau=1}^{t-1} \sigma_i(x_i^\tau,a^{\tau-1})(a_{i\tau})\Big]}\cdot
\end{align*}
This is well-defined because $a^{t-1}$ is consistent with $\sigma_i$. The important property that we will need, is that for any strategy profile $\sigma=(\sigma_i,\sigma_j)$ in which a given action history $a^{t-1}$ is  reached with positive probability, it holds that
\begin{align}
\label{eq:conditional-probs}
    \mathbb{P}_{\sigma}\left[a_{it}=I \mid \theta = H, a^{t-1}\right] =\sigma_i(a^{t-1})(I|H).
\end{align}
Hence, the conditional probability $\mathbb{P}_{\sigma}\left[a_{it}=I \mid \theta = H, a^{t-1}\right]$ is independent of $\sigma_j$. We show this formally in the appendix (see Claim \ref{claim:equivcondition}).

\begin{definition}
Fix a cool-off function $T$ and a constant $\varepsilon > 0$. We define the set $\Sigma_i(T,\varepsilon) \subset \Sigma_i(T)$ as the set of all player $i$'s \emph{cool-off strategies} $\sigma_i$ with the following properties:
\begin{enumerate}
\item $\sigma_i \in \Sigma_i(T)$.
\item In an investment phase, player $i$ invests with probability at least $1-\varepsilon$, conditioned on the state $\theta = H$. That is, for any action history $a^{t-1}\in \mathcal{I}$ which is consistent with $\sigma_i$,
\begin{align}
    \sigma_i(a^{t-1})(I|H) \geq 1-\varepsilon. 
\label{eq:coolstrat}
\end{align}
\end{enumerate}
\end{definition}
As we noted above in \eqref{eq:conditional-probs}, the second condition is equivalent to
\begin{align*}
    \mathbb{P}_{\sigma}\left[a_{it}=I \mid \theta = H, a^{t-1}\right] \geq 1-\varepsilon.
\end{align*}

The following claim is again straightforward to prove.

\begin{claim}\label{clm:legal-closed-convex}
Fix a cool-off function $T$ and a constant $\varepsilon > 0$. Then, the set $\Sigma_i(T,\varepsilon)$ is a closed and convex subset of the set $\Sigma_i$. 
\end{claim}

\subsection{Cool-off Equilibria as Fixed Points}
We have introduced the key class of strategies.
In what follows, we will find an equilibrium in this class. To do this, we make an appropriate choice of a cool-off function $T$ and an $\varepsilon$, and restrict the players to choose strategies from $\Sigma_i(T,\varepsilon)$. Since this set is closed and convex (Claim~\ref{clm:legal-closed-convex}), we can appeal to a standard fixed point argument to show the existence of an equilibrium. We next show that this equilibrium in restricted strategies is, in fact, still an equilibrium when the restriction is removed. Finally, we show that in this equilibrium, players will eventually coordinate on the right action. Moreover, when the players are very patient, the equilibrium achieves close to efficient payoffs.

\begin{lemma}[Continuation Values]\label{lem:pure}
Let $\sigma_j$ be any mixed strategy of player $j$. Fix an action history $a^{t-1}$. Then the continuation value of player $i$ at the action history $a^{t-1}$ is a non-decreasing convex function of his belief in the state $\theta = H$.
\end{lemma}

\begin{lemma}[Threshold Rule]\label{lem:threshold}
For each $c, \delta \in (0,1)$, let $T_{0} \in \mathbb{N}$ be large enough and $\varepsilon > 0$ small enough that
\begin{align}
    \label{eq:threshold}
    \delta^{T_0} < (1-\delta)(1-c-\varepsilon).
\end{align}
Let a cool-off function $T$ be such that $T(t) \ge T_0$ for all $t\geq 2$.
Fix player $j$'s strategy $\sigma_j \in \Sigma_j(T,\varepsilon)$ and an action history $a^{t-1} \in \mathcal{I}$.
Then, for every best response of player $i$, there exists a threshold $\pi$ such that at action history $a^{t-1}$ player $i$ invests (resp.\ does not invest) if her belief in the state $\theta = H$ is above (resp.\ below) the threshold $\pi$.
\end{lemma}

Lemma \ref{lem:threshold} implies that for an appropriate choice of $T$ and $\varepsilon$, it holds in every equilibrium of the game restricted to $\cS(T,\varepsilon) = \Sigma_1(T,\varepsilon) \times \Sigma_2(T,\varepsilon)$ that both players play threshold strategies: in each history each player has a threshold such that he invests above it and does not below it. Thus, when a player observes another player investing, he updates upwards (in the sense of stochastic dominance) his belief regarding the other player's signal. A consequence of this is the following corollary, which states this idea formally.

In what follows, we shall denote by $p_{it}$ the private belief of player $i$ at period $t$ before choosing action $a_{it}$. The belief $p_{it}$ is computed based on player $i$'s private history $h_{it} = (a^{t-1},x_{i}^{t})$. Hence, $p_{it} = \mathbb{P}[\theta = H \mid h_{it}]$. Finally, for convenience, we shall sometimes write the expectation operator $\mathbb{E}_{\sigma}$ and probability function $\mathbb{P}_{\sigma}$ without the subscript $\sigma$ simply as $\mathbb{E}$ and $\mathbb{P}$.

\begin{corollary}[Monotonicity of Expected Beliefs]\label{cor:monotone}
For each $c, \delta \in (0,1)$, let $T_{0} \in \mathbb{N}$ be large enough and $\varepsilon > 0$ small enough so as to satisfy condition \eqref{eq:threshold} of Lemma \ref{lem:threshold}. Let $T$ be a cool-off function such that $T(t)\geq T_0$ for all $t \geq 2$. Then, in every equilibrium of the game restricted to $\Sigma_1(T,\varepsilon) \times \Sigma_2(T,\varepsilon)$ it holds that if $a^{s-1}$ extends $a^{t-1}$ with both players investing in periods $t,t+1,\ldots,s-1$, then
\begin{align*}
    \mathbb{E}\bigl[p_{i s} \mid \theta=H, a^{s-1}\bigr]
\ge \mathbb{E}\bigl[p_{i t} \mid \theta=H, a^{t-1}\bigr].
\end{align*}
\end{corollary}

The next two propositions are the key ingredient in showing that an equilibrium in our restricted game is also an equilibrium of the unrestricted game.
\begin{proposition}\label{prop:legal-best-response}
For each $c,\delta \in (0,1)$, take $T_0 \in \N$ large enough and $\varepsilon > 0$ small enough such that
\begin{align}
    \label{eq:legal-best-response}
    \frac{c}{(1-\varepsilon)-(1-c)\delta^{T_0}/(1-\delta)}
<   c(1+2\varepsilon)
<   1-\varepsilon.
\end{align}
Let a cool-off function $T: \N \rightarrow \N$ be such that $T(t) \ge T_0$ for all $t \in \N$. Fix player $j$'s strategy $\sigma_j \in \Sigma_j(T,\varepsilon)$ and an action history $a^{t-1} \in \mathcal{I}$. Then, in every best response of player $i$, if
\begin{align*}
    \mathbb{E}\left[p_{i t} \mid \theta = H, a^{t-1}\right]
>   1-\varepsilon^{2},
\end{align*}
then player $i$ invests with probability at least $1-\varepsilon$ in the state $\theta = H$. Moreover, player $i$ invests if his private belief is above $c(1+2\varepsilon)$.
\end{proposition}

\begin{proof}
Denote by $p$ the supremum of the beliefs in which player $i$ does not invest at the action history $a^{t-1}$. Given the belief $p$, the payoff from investing is at least $(1-\delta)(-c+p(1-\varepsilon))$, while the payoff from not investing is at most $p\delta^{T_0}(1-c)$. Thus, we have that $(1-\delta)(-c+p(1-\varepsilon)) \leq p\delta^{T_0}(1-c)$. Rearranging the terms, we have
\begin{align*}
    p
\le \frac{c}{(1-\varepsilon)-(1-c)\delta^{T_0}/(1-\delta)},
\end{align*}
which implies that $p < c(1+2\varepsilon)$ by condition \eqref{eq:legal-best-response}.
Hence, player $i$ invests if his private belief $p_{it}$ is above $c(1+2\varepsilon)$.

It follows from Markov's inequality that
\begin{align*}
	\mathbb{P}(p_{it} \le p \mid \theta=H, a^{t-1})
=	\mathbb{P}(1-p_{it} \ge 1-p \mid \theta=H, a^{t-1})
\le	\frac{\mathbb{E}[1-p_{it} \mid \theta=H, a^{t-1}]}{1-p}.
\end{align*}
Since $p < c(1+2\varepsilon)$ and $\mathbb{E}[p_{it} \mid \theta=H, a^{t-1}] > 1-\varepsilon^2$,
\begin{align*}
	\frac{\mathbb{E}[1-p_{it} \mid \theta=H, a^{t-1}]}{1-p}
<	\frac{1-(1-\varepsilon^2)}{1-c(1+2\varepsilon)}
<   \varepsilon,
\end{align*}
where the latter inequality follows by condition \eqref{eq:legal-best-response}.
Hence, $\mathbb{P}(p_{it} \le p \mid \theta=H, a^{t-1}) \leq \varepsilon$. That is, player $i$ invests with probability at least $1-\varepsilon$.
\end{proof}

\begin{proposition}\label{prop:legal-best-response2}
For each $c,\delta \in (0,1)$, let $T_{0} \in \N$ and $\varepsilon >0$ be such that condition \eqref{eq:threshold} in Lemma \ref{lem:threshold} is satisfied. Then, there is a cool-off function $T: \N \to \N$ with $T(t) \geq T_0$ for each $t \geq 2$ such that, in every equilibrium strategy profile $\sigma \in \Sigma_1(T,\varepsilon) \times \Sigma_2(T,\varepsilon)$, it holds in every action history $a^{t-1} \in \mathcal{I}$ that for each $i=1,2$,
\begin{align*}
    \mathbb{E}_{\sigma}\bigl[p_{i t} \mid \theta = H, a^{t-1}\bigr]   > 1-\varepsilon^2.
\end{align*}
\end{proposition}

\begin{proof}
Fix $T_{0} \in \N$ and $\varepsilon > 0$. Since private signals are bounded, for each $s \geq 1$, there exists $\underline{p}_{s} > 0$ such that for each $i=1,2$,
\begin{align*}
\mathbb{P}[p_{is} > \underline{p}_s] = 1.
\end{align*}
irrespective of the strategies of the players. For example, $\underline{p}_{s}$ could be defined as the belief that a player would have if he observed both players received the worst signal repeatedly up to period $s$, where the worst signal is defined as the one which has the lowest likelihood ratio. Note that $\underline{p}_s$ only depends on the period $s$ and not on the action history or the private history of any player.

We now construct the cool-off function $T$. For each $s \geq 1$, define $T(s) \geq T_0$ large enough to be such that if a cool-off phase were to begin at period $s$ and a player's belief were $\underline{p}_{s}$, then when the cool-off phase ends at time period $s + T(s) - 1$, the expectation of a player's private belief conditional on $\theta = H$ would be above $1-\varepsilon^{2}$.

For this cool-off function $T$, from Corollary \ref{cor:monotone} it follows that for each equilibrium strategy profile $\sigma \in \Sigma_1(T,\varepsilon) \times \Sigma_2(T,\varepsilon)$ and action history $a^{t-1} \in \mathcal{I}$, we have that
\begin{align*}
    \mathbb{E}_{\sigma}\bigl[p_{i t} \mid \theta = H, a^{t-1}\bigr]
>   1-\varepsilon^{2},
\end{align*}
which completes the proof.
\end{proof}

Now we prove the existence of equilibria in cool-off strategies.
Applying a fixed-point theorem to the game restricted to $\Sigma_{1}(T,\varepsilon) \times \Sigma_{2}(T,\varepsilon)$, we find an equilibrium $\sigma^{\ast} \in \cS(T,\varepsilon)$.
We then show that this equilibrium remains an equilibrium of the unrestricted game, using Propositions \ref{prop:legal-best-response} and \ref{prop:legal-best-response2}.

\begin{proposition}\label{prop:fixedpoint}
For each $c,\delta \in (0,1)$, let $T_0$ be large enough and $\varepsilon$ small enough so as to satisfy condition \eqref{eq:threshold} in Lemma \ref{lem:threshold} and condition \eqref{eq:legal-best-response} in Proposition \ref{prop:legal-best-response}. Then, there exists a cool-off function $T$ with $T(t) \geq T_0$ for all $t \geq 2$ and there exists an equilibrium $\sigma^* = (\sigma_i^*,\sigma_j^*)$ such that $\sigma_i^* \in \Sigma_i(T,\varepsilon)$ and $\sigma_j^* \in \Sigma_j(T,\varepsilon)$. 
\end{proposition}

\subsection{Efficiency}
\begin{theorem}[Action Efficiency]\label{thm:efficiency}
For each $c,\delta \in (0,1)$, there exist a cool-off function $T$, a constant $\varepsilon > 0$ such that for any equilibrium $\sigma^* \in \cS(T,\varepsilon)$, players eventually choose the right actions. That is, the equilibrium action profile converges to $(I,I)$ in the state $\theta=H$ and to $(N,N)$ in the state $\theta=L$ almost surely.
\end{theorem}

\begin{proof}
The underlying  probability space for defining the events of interest will be the set of all outcomes $\Omega = \Theta \times A^{\infty}\times (X^2)^{\infty}$. Let $\sigma^* \in \cS(T,\varepsilon)$ be any equilibrium. Define the events
\begin{align*}
A^{\infty}_I = \bigl\{ \omega \in \Omega : \lim_{t\rightarrow \infty} a_t(\omega) = (I,I) \bigr\}, \quad
A^{\infty}_N = \bigl\{ \omega \in \Omega : \lim_{t\rightarrow \infty} a_t(\omega) = (N,N) \bigr\}.
\end{align*}
We wish to show that $\mathbb{P}(A^{\infty}_I \mid \theta = H) = \mathbb{P}(A^{\infty}_N \mid \theta = L) = 1$. Consider the latter equality. It would suffice to show that there exists a $\underline{p} \in (0,1)$ such that no player invests when his private belief is below $\underline{p}$. In the state $\theta = L$, players' private beliefs almost surely converge to $0$ (since they can learn from their own signals) and hence will eventually fall below $\underline{p}$, leading both players to not invest from some point on. We derive such a $\underline{p}$.

When a player has belief $p$, the payoff from investing is at most $(1-\delta) (p-c) + \delta (p(1-c))$ and the payoff from not investing is at least $0$. Hence, if a player finds it optimal to invest at $p$, then
\begin{align*}
    p \geq \displaystyle\frac{c-c\delta}{1-c\delta}.
\end{align*}
Hence, we can define $\underline{p} = \frac{c-c\delta}{1-c\delta}$. It now follows that $\mathbb{P}(A^{\infty}_N \mid \theta = L) = 1$.

We next show $\mathbb{P}(A^{\infty}_I \mid \theta = H) = 1$. Consider the event $B^{\infty} = \{ \omega \in \Omega : \lim_{t\rightarrow \infty} p_{it}(\omega) = 1 \mbox{ for both } i=1,2 \}$. We know that
$\mathbb{P}(B^{\infty} \mid \theta = H) = 1$. Hence, it suffices to show that $B^{\infty} \subseteq A^{\infty}_I$.

Let $\omega \in B^{\infty}$. It suffices to show that on the outcome path induced by $\omega$, there must only be finitely many cool-off periods. Suppose not. Then, there exist infinitely many cool-off periods in $\omega$. Recall from Proposition \ref{prop:legal-best-response} that a player invests with probability one in an investment phase when his private belief is above $\bar{p} = c(1+2\varepsilon)$. Now, since $\omega \in B^{\infty}$, there exists a $\bar{T}$ such that $p_{it}(\omega) > \bar{p}$ for all $t \geq \bar{T}$ for both $i=1,2$. Now, let $t'$ be the starting period of the first cool-off phase after $\bar{T}$. Then, $T(t') + t' \geq \bar{T}$. Hence, when this cool-off period ends at time $s= T(t') + t' - 1$, we enter an investment phase. Since this investment phase at $s+1$ starts beyond period $T'$, the private beliefs of both players would be above $\bar{p}$ and both players would invest in time period $s+1$. Next period, players would continue to be in the investment phase and again invest. Proceeding this way, players would keep on investing which means the investment phase would last forever. This contradicts the fact that there are infinitely many cool-off periods. Hence, there exist only finitely many cool-off periods in $\omega$, which means that players invest forever after some point of time. Hence, $\omega \in A^{\infty}_I$. 
\end{proof}

\begin{theorem}[Payoff Efficiency]
Fix any $c \in (0,1)$. For every $\Delta > 0$, there exists $\bar{\delta} \in (0,1)$ such that for each $\delta \ge \bar{\delta}$, there exist a cool-off function $T$, a constant $\varepsilon > 0$, such that in any equilibrium $\sigma^* \in \cS(T,\varepsilon)$, both players obtain a payoff of at least $1-c-\Delta$ in the state $\theta=H$ and a payoff of at least $-\Delta$ in the state $\theta=L$.
\end{theorem}
\begin{proof}
Fix $c,\Delta \in (0,1)$. Let $\varepsilon >0$ be small enough so that $c < c(1+2\varepsilon)(1-\varepsilon) < (1-\varepsilon)^2$ and moreover
\begin{align}
    (1-\varepsilon)^2
    \left(
        \left(1 - \frac{2\varepsilon}{1-\bar{p}}\right)(1-c) + \frac{2\varepsilon}{1-\bar{p}}(-c) 
    \right) > 1-c-\Delta. \label{payoffeffH}
\end{align}
and 
\begin{align}
    (1-\varepsilon)\Big(\displaystyle\frac{-\Delta}{2}\Big) -\varepsilon c > -\Delta. \label{payoffeffL}
\end{align}
where recall from Proposition \ref{prop:legal-best-response} that $\bar{p} = c(1+2\varepsilon)$. In what follows, we will show that players achieve a payoff equal to the LHS of \eqref{payoffeffH} in state $\theta = H$ and a payoff equal to the LHS of \ref{payoffeffL} in state $\theta = L$. Let $T_1 \in \mathbb{N}$ be such that for any strategy profile which  implements a cool-off scheme $T$ with $T(1) = T_1$, it holds that i) $\mathbb{E}[p_{iT_1} \mid \theta = H] > 1- \varepsilon^2$ for all $i = 1,2$; ii) $\mathbb{P}[p_{iT_1} > 1-\varepsilon \mbox{ for each }  i=1,2 \mid \theta = H] > 1-\varepsilon$ ; iii) $\mathbb{P}[p_{iT_1} < \frac{\Delta}{2(1-c)+\Delta} \mbox{ for each }  i=1,2 \mid \theta = H] > 1-\varepsilon$ . Now let $\bar{\delta} \in (0,1)$ be such that $\bar{\delta}^{T_1} > 1-\varepsilon$. Let $\delta > \bar{\delta}$. Now, suppose that $T_0$ is such that condition \eqref{eq:threshold} in Lemma \ref{lem:threshold} and condition \eqref{eq:legal-best-response} in Proposition \ref{prop:legal-best-response} are satisfied for $T_0$ and $\varepsilon$. From Proposition \ref{prop:fixedpoint}, it follows that there exists a cool-off function with $T(1) = T_1$ and $T(t) \geq T_0$ for all $t \geq 2$ and there exists an equilibrium of the repeated game restricted to strategies in $\cS(T,\varepsilon)$. Let $\sigma^*$ be one such equilibrium.

We shall show that in this equilibrium each player obtains a long-run expected payoff of at least $1-c-\Delta$ in the state $\theta = H$ and $-\Delta$ in state $\theta = L$. We first establish the former. Note that we have $1-\varepsilon > c(1+2\varepsilon)$. Hence, if the private beliefs of both players are above $1-\varepsilon$ at the end of the first cool-off phase with length $T_{1}$, and stay above $1-\varepsilon$ forever, both players would invest forever. We shall show that the probability with which this happens is at least $1 - \frac{2\varepsilon}{1-\bar{p}}$. We argue as follows. Suppose that private beliefs for players $1$ and $2$ are $p_1,p_2 \geq 1-\varepsilon$ at time $T_1$. Hence, from period $T_1$ onwards, the probability that both players invest forever is at least 
\begin{align}\label{infinv}
\mathbb{P}\bigl[p_{it} > \bar{p} \textup{ for each $i=1,2$ and $t \geq T_{1}$} \mid p_{iT_1} = p_i, p_{jT_1} = p_j, \theta = H\bigr].
\end{align}
Note that conditional on $p_{iT_1} = p_i, p_{jT_1} = p_j$ and $ \theta = H$, the belief process $p_{it}$ is a bounded submartingale. Now define the stopping time 
\begin{align*}
T_i(\omega) = \min \{t \mid p_{it}(\omega) \leq \bar{p}\}.
\end{align*}
Let $p'_{it}$ be the process $p_{it}$ stopped at time $T_i$. It follows that $p'_{it}$ is a bounded submartingale as well. From the Martingale Convergence Theorem, it converges almost surely to a limit $p'_{i,\infty}$. Hence, the conditional probability in \eqref{infinv} is at least 
\begin{align*}
\mathbb{P}\bigl[p'_{i,\infty} > \bar{p} \textup{ for each $i=1,2$}  \mid p_{iT_1} = p_i, p_{jT_1} = p_j, \theta = H\bigr].
\end{align*}
From the Optional Stopping Theorem, we have that
\begin{align*}
 \mathbb{E}\bigl[p'_{i,\infty} \mid p_{iT_1} = p_i, p_{jT_1} = p_j, \theta = H\bigr] \geq p_i.
\end{align*}
Since $1-p'_{i,\infty}$ is non-negative and $p_i \geq 1-\varepsilon$, from Markov's inequality it follows that for each $i$,
\begin{align*}
 \mathbb{P}\bigl[p'_{i,\infty} \leq \bar{p} \mid p_{iT_1} = p_i, p_{jT_1} = p_j, \theta = H\bigr]  \leq \displaystyle\frac{\varepsilon}{1-\bar{p}}.
\end{align*}
From the above we obtain that the probability that both players will invest from period $T_1$ onwards is at least $1 - \frac{2\varepsilon}{1-\bar{p}}$. Hence, the long-run expected payoff to each player in state $\theta = H$ will be at least equal to the LHS of \eqref{payoffeffH} which is greater than $1-c-\Delta$.

Finally, we show that under $\sigma^*$, the payoff in state $L$ is at most $-\Delta$. To see this, we first make the following observation.  Suppose at the end of the first $T_1$ cool-off periods, a player's private belief is $p \leq \frac{\Delta}{2(1-c) + \Delta}$. Now from Lemma \ref{lem:pure}, we know that a player's optimal continuation value $v(\cdot)$ is a non-decreasing convex function of his private belief $p$. We show that if the optimal continuation strategy of the player yields payoffs $(v_H,v_L)$ at $p$, then $v_L \geq -\frac{\Delta}{2}$. Since $(v_H,v_L)$ is the optimal continuation value at $p$, it follows that the vector $(v_H,v_L)$ supports the convex function $v(\cdot)$ at $p$. Hence, 
\begin{align*}
v(p) = p v_H + (1-p) v_L.
\end{align*}
Note that $v(p) \geq 0$ since a player can always guarantee himself a payoff of zero forever and $v_H \leq 1-c$ since the best payoff a player can get in state $H$ is $1-c$. These, together with the fact that $p \leq \frac{\Delta}{2(1-c) + \Delta}$, imply $v_L \geq -\frac{\Delta}{2}$. Now, since we have $\mathbb{P}[p_{iT_1} < \frac{\Delta}{2(1-c)+\Delta} \mbox{ for each }  i=1,2 \mid \theta = H] > 1-\varepsilon$, the ex-ante payoff in state $L$ to each player is at least equal to the LHS of \eqref{payoffeffL}, which is greater than $-\Delta$. 
\end{proof}

\section{Discussion}
\label{section: discussion}

\subsection{Ex-post Equilibria}
Most of the literature in repeated games assumes that the set of
feasible and individually rational payoffs is full-dimensional.  That
is, the dimension of the set is equal to the number of players.
Analogously, existing work in repeated games of incomplete information
assumes a version of full-dimensionality, called statewise
full-dimensionality \citep[e.g.,][]{wiseman2005partial,
  wiseman2012partial, yamamoto, sugayayamamoto}.  However, our payoff structure does not satisfy this condition.  In addition, two closely related papers to ours,
\cite{yamamoto} and \cite{sugayayamamoto}, focus on a special
class of sequential equilibria, called ex-post equilibria.  Yet, any
ex-post equilibrium is inefficient in our model, and this is due to
the lack of the statewise full-dimensionality.  We discuss each of
these two differences in turn.

\paragraph{Statewise Full-dimensionality.}
For each state $\theta$, we define the set of feasible set of payoffs by
\begin{align*}
	W(\theta)
=	\textup{co}\,
	\Bigl(
	\Bigl\{
		v \in \mathbb{R}^{2} : 
		\exists\, a \in A\quad \forall\, i \in \mathcal{N}\quad u_{i}(a, \theta) = v_{i}
	\Bigr\}
	\Bigr).
\end{align*}
Further, we define player $i$'s minimax payoff in state $\theta$ by
\begin{align*}
	\underline{w}_{i}(\theta)
=	\min_{\alpha_{j} \in \Delta A_{j}} \max_{a_{i} \in A_{i}} u_{i}(a_{i}, \alpha_{j}, \theta)
\end{align*}
and the set of all feasible and individually rational payoffs in state $\theta$ by
\begin{align*}
	W^{\ast}(\theta)
=	\Big\{
		w \in W(\theta) : 
		\forall\, i \in \mathcal{N}\quad w_{i} \ge \underline{w}_{i}(\theta)
	\Big\}.
\end{align*}

Statewise full-dimensionality is said to be satisfied if the set $W^{\ast}(\theta)$ is of dimension equal to the number of players for each $\theta$.
\footnote{When there is no uncertainty of payoff-relevant states,  statewise full-dimensionality is equivalent to the full-dimensionality in the literature of repeated games of complete information.}
The statewise full-dimensionality is assumed in most of the literature of repeated games of incomplete information, but it is not in our model.
As drawn in Figure \ref{fig:statewisefull-dimensionality}, the set $W^{\ast}(H)$ is of dimension $2$, but the set $W^{\ast}(L) = \{(0,0)\}$ is not.
Intuitively speaking, this lack of statewise full-dimensionality means that players' incentives, thus behavior, are quite different from state to state.
This could make it difficult to design an appropriate punishment scheme that supports cooperation.

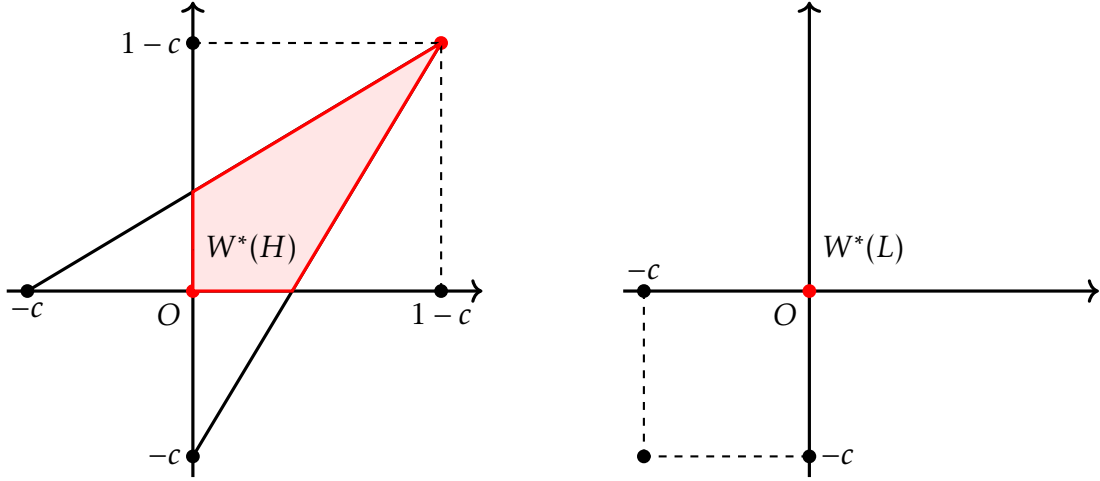
\begin{figure}[t]
\centering
\begin{minipage}{0.49\hsize}
	\centering
	\begin{tikzpicture}[scale=0.55,thick]
	\draw[->,very thick] (-4.5,0) -- (7,0) node[right] {};
	\draw[->,very thick] (0,-4.5) -- (0,7) node[above] {};
	\draw[red,fill=red,thick] (0,0) circle (4pt) node[] {};
	\draw[] (0,0) node[below left] {$O$};
	\draw[red,fill=red,thick] (6,6) circle (4pt) node[] {};
	\draw[black,fill=black,thick] (6,0) circle (4pt) node[below] {$1-c$};
	\draw[black,fill=black,thick] (0,6) circle (4pt) node[left] {$1-c$};
	\draw[black,fill=black,thick] (-4,0) circle (4pt) node[below] {$-c$};
	\draw[black,fill=black,thick] (0,-4) circle (4pt) node[left] {$-c$};
	\draw[black,very thick] (-4,0) coordinate -- (6,6) coordinate (E);
	\draw[black,very thick] (0,-4) coordinate -- (6,6) coordinate (E);
	\draw[dashed,thick] (6,0) coordinate -- (E);
	\draw[dashed,thick] (0,6) coordinate -- (E);
	\filldraw[draw=red,fill=red!10!white,very thick] (0,0) -- (2.4,0) -- (E) -- (0,2.4) -- (0,0);
	\draw[black,fill=black,thick] (0.0,1.0) circle (0pt) node[right] {$W^{\ast}(H)$};
	\end{tikzpicture}
\end{minipage}
\begin{minipage}{0.49\hsize}
	\centering
	\begin{tikzpicture}[scale=0.55,thick]
	\draw[->,very thick] (-4.5,0) -- (7,0) node[right] {};
	\draw[->,very thick] (0,-4.5) -- (0,7) node[above] {};
	\draw[red,fill=red,thick] (0,0) circle (4pt) node[] {};
	\draw[] (0,0) node[below left] {$O$};
	\draw[black,fill=black,thick] (-4,-4) circle (4pt) node[] {};
	\draw[black,fill=black,thick] (-4,0) circle (4pt) node[above] {$-c$};
	\draw[black,fill=black,thick] (0,-4) circle (4pt) node[right] {$-c$};
	\draw[dashed,thick] (-4,0) coordinate -- (-4,-4);
	\draw[dashed,thick] (0,-4) coordinate -- (-4,-4);
	\draw[black,fill=black,thick] (0.0,1.0) circle (0pt) node[right] {$W^{\ast}(L)$};
	\end{tikzpicture}
\end{minipage}
\caption{the Lack of Statewise Full-dimensionality}
\label{fig:statewisefull-dimensionality}
\end{figure}

\paragraph{Ex-post Equilibria.}
Two closely related papers to ours, \cite{yamamoto} and \cite{sugayayamamoto}, focus on a special class of sequential equilibria,
called \emph{ex-post equilibria}. An ex-post equilibrium is defined as
a sequential equilibrium in the infinitely repeated game in which the
strategies are such that they would be optimal even if the state
$\theta$ were common knowledge. In an ex-post equilibrium, player
$i$'s continuation play after history $h_{it}$ is optimal regardless
of the state $\theta$ for each $h_{it}$ (including off-path
histories).

No ex-post equilibrium, however, can approximate efficiency in our model. More precisely, in any ex-post equilibrium, neither player invests at any private history. 
To see this, note that in any ex-post equilibrium, a player's strategy would be optimal if the state $L$ were common knowledge. Since in the state $L$, a player's best response to any strategy of the opponent is to not invest at any private history, the player never invests even in the state $H$. Thus, neither player ever invests in any ex-post equilibrium.

Ex-post equilibria may yield a trivial play if the payoff structure does not satisfy statewise full-dimensionality. Further, they may not exist if a stage game is like a global game. To see this, suppose, for example, that there is another state in which irrespective of the opponent's action, a player gains payoff $1-c$ if he invests and payoff of $0$ if he does not. Then, players always invest if that state were common knowledge, but they would never invest if the state $L$ were common knowledge; therefore, there would not exist ex-post equilibria.

\subsection{Common Learning} 
A coordination game which is very similar to ours is presented as one
of the motivations for the introduction of the concept of common
learning by~\cite{cems1}.  We now explain this concept, and show that,
in our game, common learning is attained in every strategy profile. 

Recall $\Omega := \Theta \times A^{\infty} \times (X^2)^{\infty}$ is the set of all possible outcomes. Any strategy profile $\sigma = (\sigma_i,\sigma_j)$ induces a probability measure $\mathbb{P}$ on $\Omega$ which is endowed with the $\sigma$-algebra $\mathcal{F}$ generated by cylinder sets. For each $F \in \mathcal{F}$ and $q \in (0,1)$, define the event $B^{q}_{it}(F) = \{\omega\in \Omega : \mathbb{P}(F \mid a^{t-1},x_i^t) \geq q\}$. Now, define the following events: 
\begin{enumerate}
\item
$B^{q}_{t}(F) = B^{q}_{it}(F) \cap B^{q}_{jt}(F)$. That is, this is the event that both players assign probability at least $q$ to the event $F$ on the basis of their private information.
\item
$[B^{q}_{t}]^n(F) =
B^{q}_{t}(B^{q}_{t}(B^{q}_{t}...B^{q}_{t}(F)))$. That is, this is the
event that both players assign probability at least $q$ that both players assign probability at least $q$ that... both players assign probability at least $q$ to the event $F$ on the basis of their private information.
\item
$C^{q}_{t}(F) = \bigcap_{n\geq 1} [B^{q}_{t}]^n(F)$.
\end{enumerate}
We say that players have \textit{common $q$-belief} in event F at $\omega$, at time $t$, if $\omega \in C^{q}_{t}(F)$. We say that players \textit{commonly learn} $\theta \in \Theta$ if for all $q\in (0,1)$, there exists a $T$ such that for all $t \geq T$,
\begin{align*}
\mathbb{P}\bigl[C^q_t(\theta) \mid \theta\bigr] > q.
\end{align*}

\begin{proposition}\label{prop:commonlearning}
Let $\sigma$ be any strategy profile which induces the probability measure $\mathbb{P}$. Then, under $\mathbb{P}$, players commonly learn the state $\theta$ for each $\theta \in \Theta$.
\end{proposition}

\section{Conclusion}
\label{section: conclusion}

This paper has attempted to answer the question of whether rational
Bayesian players in a coordination game, each of whom is receiving
private signals about the unknown state of nature, can both learn to
play the action appropriate to the true state and do so in a manner
such that the loss in payoffs becomes negligible as the discount
factor goes to one. The answer, in a simple setup, is affirmative,
even though the particular coordination game we consider does not have
properties such as statewise full dimensionality. Future research can
be considered along several dimensions. The obvious extension, which
is most likely straightforward, would be to have an arbitrary finite
number of players and states of nature. A more interesting question is
whether our results still hold whenever the private signal structure
is such that players always commonly learn the state. Costly
information acquisition is another possible extension. Finally,
another interesting follow-up to this work would be to see if a
similar set of results can be obtained for boundedly rational players
who do not use full Bayesian updating but are constrained by memory
restrictions.

\appendix
\section{Appendix}
\subsection{Strategies}
\begin{claim} (Equivalence of conditions) 
Let $\sigma_i \in \mathcal{S}(T,\varepsilon)$ and let $\sigma_j$ be any strategy of player $j$. Now, suppose $a^{t-1}$ is reached with positive probability under $\sigma=(\sigma_i,\sigma_j)$. Then, 
\begin{align}
    \mathbb{P}_{\sigma}\left[a_{it}=I \mid \theta = H, a^{t-1}\right] =\sigma_i(a^{t-1})(I|H).
\label{claim:equiv}
\end{align}
Hence, the conditional probability $\mathbb{P}_{\sigma}\left[a_{it}=I \mid \theta = H, a^{t-1}\right]$ is independent of $\sigma_j$. 
\label{claim:equivcondition}
\end{claim}
\begin{proof}
Firstly note that since $a^{t-1}$ is reached with positive probability under $\sigma$, it is the case that $a^{t-1}$ is consistent with $\sigma_i$. Hence, the conditional probability in \ref{claim:equiv} is well-defined.
\begin{align*}
\mathbb{P}_{\sigma}\left[a_{it}=I \mid \theta = H, a^{t-1}\right] = \frac{\mathbb{P}_{\sigma}\left[a_{it}=I, \theta = H, a^{t-1}\right]}{\mathbb{P}_{\sigma}\left[\theta = H, a^{t-1}\right]},
\end{align*}
The numerator is equal to 
\begin{align*}
 \displaystyle\sum_{x^t_i \in X^t} \displaystyle\sum_{x^t_j \in X^t} f^H(x^t_i) f^H(x^t_j) \Big[\prod_{\tau=1}^{t-1} \sigma_i(x_i^\tau,a^{\tau-1})(a_{i\tau})\big(\sigma_i(x^t_i,a^{t-1})(I)\big)\Big]\Big[\prod_{\tau=1}^{t-1} \sigma_j(x_j^\tau,a^{\tau-1})(a_{j\tau})\Big] \\  
 = 
 \displaystyle\Big[\sum_{x^t_i \in X^t} f^H(x^t_i)\prod_{\tau=1}^{t-1} \sigma_i(x_i^\tau,a^{\tau-1})(a_{i\tau})\big(\sigma_i(x^t_i,a^{t-1})(I)\big)\Big]\Big[\sum_{x^t_j \in X^t}  f^H(x^t_j)\prod_{\tau=1}^{t-1} \sigma_j(x_j^\tau,a^{\tau-1})(a_{j\tau})\Big].
\end{align*}
The equality follows from the fact that signals are conditionally independent across players, which means that the probability that players receive signals $(x^t_i,x^t_j)$ is equal to $f^H(x^t_i) f^H(x^t_j)$. By the same token, the denominator is equal to
\begin{align*}
\displaystyle\sum_{x^t_i \in X^t} \displaystyle\sum_{x^t_j \in X^t} f^H(x^t_i) f^H(x^t_j)\Big[\prod_{\tau=1}^{t-1} \sigma_i(x_i^\tau,a^{\tau-1})(a_{i\tau})\Big]\Big[\prod_{\tau=1}^{t-1} \sigma_j(x_j^\tau,a^{\tau-1})(a_{j\tau})\Big]. 
 \\  
 = \displaystyle \Big[\sum_{x^t_i \in X^t}f^H(x^t_i) \prod_{\tau=1}^{t-1} \sigma_i(x_i^\tau,a^{\tau-1})(a_{i\tau})\Big]\Big[\sum_{x^t_j \in X^t} f^H(x^t_j)\prod_{\tau=1}^{t-1} \sigma_j(x_j^\tau,a^{\tau-1})(a_{j\tau})\Big].
\end{align*}
Now, notice that since the denominator is strictly positive, it follows that the term 
$$\sum_{x^t_j \in X^t} f^H(x^t_j)\prod_{\tau=1}^{t-1} \sigma_j(x_j^\tau,a^{\tau-1})(a_{j\tau})$$ 
in the denominator is also strictly positive. This terms also cancels out from the numerator and the denominator and we obtain 
  \begin{align*}
    \mathbb{P}_{\sigma}\left[a_{it}=I \mid \theta = H, a^{t-1}\right] = \frac{\displaystyle\sum_{x^t_i \in X^t} f^H(x^t_i) \Big[\prod_{\tau=1}^{t-1} \sigma_i(x_i^\tau,a^{\tau-1})(a_{i\tau})\big(\sigma_i(x^t_i,a^{t-1})(I)\big)\Big]}{\displaystyle\sum_{x^t_i \in X^t} f^H(x^t_i) \Big[\prod_{\tau=1}^{t-1} \sigma_i(x_i^\tau,a^{\tau-1})(a_{i\tau})\Big]}\cdot
\end{align*}
Hence, $\mathbb{P}_{\sigma}\left[a_{it}=I \mid \theta = H, a^{t-1}\right] =\sigma_i(a^{t-1})(I|H).$
\end{proof}

\subsection{Two Lemmas}
The following lemmas shall be useful in our analysis. The first lemma states that beliefs in an event---which are a martingale---form a submartingale when conditioned on the same event.
\begin{lemma} \label{lem:condexpineq}
Let $\mathcal{F}_1 \subseteq \mathcal{F}_2$ be sigma-algebras in a finite probability space, let $p=\mathbb{P}[E\mid \mathcal{F}_1]$ be the prior probability of some event $E$, and let $q=\mathbb{P}[E\mid \mathcal{F}_2]$ be its posterior probability. Then, 
\begin{align*}
    \mathbb{E}[q\mid E, p] = \frac{\mathbb{E}[q^2 \mid p]}{\mathbb{E}[q \mid p]} \geq p.
    \end{align*}
\end{lemma}
\begin{proof}
Since the space is finite, we can write
\begin{align*}
    \mathbb{E}[q\mid E,p]  =\sum_{x \in [0,1]}x \cdot \mathbb{P}[q=x\mid E,p].
\end{align*}
By Bayes' law,
\begin{align*}
    =\sum_{x \in [0,1]}x \cdot \mathbb{P}[E\mid q=x,p] \cdot \frac{\mathbb{P}[q=x\mid p]}{\mathbb{P}[E\mid p]}.
\end{align*}
Note that $\mathbb{P}[E\mid p]=p$, and that since $\mathcal{F}_1 \subseteq \mathcal{F}_2$, $\mathbb{P}[E\mid q=x,p]=x$. Hence
\begin{align*}
    =\sum_{x \in [0,1]}x^2 \cdot \frac{\mathbb{P}[q=x\mid p]}{p} = \frac{1}{p}\mathbb{E}[q^2\mid p].
\end{align*}
Again using the fact that $\mathcal{F}_1 \subseteq \mathcal{F}_2$, $p = \mathbb{E}[q\mid p]$, we have shown that
\begin{align*}
    =\frac{\mathbb{E}[q^2 \mid p]}{\mathbb{E}[q \mid p]}\cdot
\end{align*}
Finally, this is at least $p$, since $\mathbb{E}[(q-p)^2\mid p] \geq 0$, and since $\mathbb{E}[q\mid p]=p$, and so $\mathbb{E}[q^2\mid p]\geq p^2$.
\end{proof}
The next lemma states that if a random variable taking values in the unit interval has a mean close to one, then it must assume values close to one with high probability. 
\begin{lemma}\label{lem:markovineq}
Let $X$ be a random variable which takes values in $[0,1]$. Suppose $\mathbb{E}[X] \geq 1 - \varepsilon^2$ where $\varepsilon\in (0,1)$. Then, 
\begin{align}
\mathbb{P}[X > 1- \varepsilon] \geq 1-\varepsilon. \nonumber 
\end{align}
\begin{proof}
Define the random variable $Y = 1-X$. Note that it suffices to show that
\begin{align}
\mathbb{P}[Y \geq \varepsilon] \leq \varepsilon. \label{markov}
\end{align}
Since $\mathbb{E}[Y] \leq \varepsilon^2$, and since $Y$ is non-negative, by Markov's inequality, we obtain (\ref{markov}).
\end{proof}
\end{lemma}

\subsection{Proofs from Section \ref{section: equilibrium analysis}}
\begin{proofof}{Claim \ref{clm:cool-offs-closed}}
The conclusion follows by definition. If $\{\sigma^n_i\}_n \subseteq \sigma^n_i$ and $\sigma^n_i \rightarrow \sigma^*_i$ in the weak-* topology. Since the support of each $\sigma^n_i$ is contained in $S_i(T)$, so is the support of $\sigma^*_i$. It also follows that any  convex combination $\sigma_i^\lambda=  \lambda \sigma_i + (1-\lambda) \sigma_i'$ of two mixed strategies $\sigma_i,\sigma_i' \in \Sigma_i(T)$ also belongs to $\Sigma_i(T)$. 
\end{proofof}

\begin{proofof}{Claim \ref{clm:legal-closed-convex}}
We first argue that $\Sigma_i(T,\varepsilon)$ is closed. Suppose $\{\sigma_i^n\} \subseteq \Sigma_i(T,\varepsilon)$ and suppose $\sigma^n_i \rightarrow \sigma_i^*$ in the weak-* topology. Firstly, note that from Claim \ref{clm:cool-offs-closed}, it follows that $\sigma^*_i \in \Sigma_i(T)$. Now suppose $a^{t-1}$ in an investment phase and is consistent with $\sigma^*_i$. Then, the denominator of $\sigma^*_i(a^{t-1})(I|H)$, as defined, will be strictly positive. This means that for large $n$, the denominator of $\sigma^n_i(a^{t-1})(I|H)$ will be strictly positive also. This means that $a^{t-1}$ is consistent with $\sigma^n_i$ for large $n$. It follows from the definition that $\sigma^n_i(a^{t-1})(I|H) \rightarrow \sigma^*_i(a^{t-1})(I|H)$ which implies that $\sigma^*_i(a^{t-1})(I|H) \geq 1-\varepsilon$.
\\ \\ 
We now show convexity. Let $\sigma_i,\sigma_i' \in \Sigma_i(T,\varepsilon)$ be two mixed strategies let $\lambda \in (0,1)$. Now, consider the mixed strategy defined as  $\sigma_i^\lambda=  \lambda \sigma_i + (1-\lambda) \sigma_i'$. Then, we show $\sigma_i^\lambda \in \Sigma_i(T,\varepsilon)$.  \\ \\ 
Clearly $\sigma^{\lambda}_i \in \Sigma_i(T)$ since $\Sigma_i(T)$ is convex from Claim \ref{clm:cool-offs-closed}. Let $\bar{\sigma_j}$ be the strategy of player $j$ that plays each action with probability $1/2$ at each private history. Denote as $\mathbb{P}^\lambda$, $\mathbb{P}$ and $\mathbb{P}'$, the probability measures induced by the strategy profiles $(\sigma_i^\lambda,\sigma_j), (\sigma_i,\sigma_j) $ and $(\sigma_i',\sigma_j)$. Note that $\mathbb{P}^\lambda = \lambda \mathbb{P} + (1-\lambda) \mathbb{P}'$. \\ \\ 
We wish to show that $\sigma_i^\lambda$ also satisfies condition \ref{eq:coolstrat}. To this end, let $a^{t-1}$ be an action history in an investment phase that is consistent with $\sigma_i^\lambda$. It follows that $a^{t-1}$ is reached with positive probability under $\sigma_i^\lambda$. Then, from Claim \ref{claim:equiv}, it suffices to show that 
\begin{align*}
    \mathbb{P}^\lambda\left[a_{it}=I \mid \theta = H, a^{t-1}\right] \geq 1-\varepsilon.
\end{align*}
Now, note by hypothesis, since $\sigma_i,\sigma_i'\in \Sigma_i(T,\varepsilon)$,  we already have that $\mathbb{P}\left[a_{it}=I, \theta = H, a^{t-1}\right] \geq (1-\varepsilon)\mathbb{P}\left[\theta = H, a^{t-1}\right]$ and $\mathbb{P}'\left[a_{it}=I, \theta = H, a^{t-1}\right] \geq (1-\varepsilon)\mathbb{P}'\left[\theta = H, a^{t-1}\right]$. These, together imply that 
\begin{align*}
\displaystyle\frac{\lambda \mathbb{P}\left[a_{it}=I, \theta = H, a^{t-1}\right] + (1-\lambda) \mathbb{P}'\left[a_{it}=I, \theta = H, a^{t-1}\right]}{\lambda \mathbb{P}\left[\theta = H, a^{t-1}\right] + (1-\lambda) \mathbb{P}'\left[\theta = H, a^{t-1}\right]} \geq 1-\varepsilon,
\end{align*}
which is equivalent to 
\begin{align*}
   \displaystyle \mathbb{P}^\lambda\left[a_{it}=I \mid \theta = H, a^{t-1}\right] = \frac{\mathbb{P}^\lambda\left[a_{it}=I,\theta = H, a^{t-1}\right]}{\mathbb{P}^\lambda\left[\theta = H, a^{t-1}\right]}\geq 1-\varepsilon.
\end{align*}
\end{proofof}

\begin{proofof}{Lemma \ref{lem:pure}}
Denote by $p$ player $i$'s belief at the action history $a^{t-1}$. For each possible continuation strategy $s$ of player $i$ (which we think of as depending only on his future signals) let $v_H$ and $v_L$ denote the expected payoff in the high and low state, respectively, so that the expected payoff for the continuation strategy $s$ is given by $p v_H+(1-p) v_L$. In any best response of player $i$, his expected payoff is given by
\begin{align*}
    v(p)
=   \max_{s} \bigl\{p v_H+(1-p) v_L\bigr\},
\end{align*}
where the maximization is over all continuation strategies for player $i$.
Since it is the maximum of convex functions, it is convex. To see that it is increasing, note that continuation value $v_L \leq 0$ for every continuation strategy $s$, that $v_L = v_H = 0$ for some continuation strategy $s$---specifically in the strategy in which the player never invests. Hence the derivative of continuation payoff $v(p)$ is at least $0$ at $p=0$, and thus continuation payoff $v(p)$, as a function of belief $p$, has a non-negative subderivative in the entire interval $[0,1]$.
\end{proofof}

\begin{figure}[t]
\centering
\begin{tikzpicture}[domain=0:6, thick]
\draw[black,fill=black] (0,0) circle (2pt) node[above left] {$O$};
\draw[very thick,->] (-1.5,0)--(7.5,0) node[below] {$p$}; 
\draw[very thick,->] (0,-1.5)--(0,5.5) node[left] {}; 

\draw[blue,very thick,domain=0:5] plot(\x,0.25*\x^1.5) node[above] {};
\draw[blue,very thick,domain=5:6] plot(\x,0.2*\x^2-2.204) node[above] {};
\draw[black,fill=black] (6,0.2*6^2-2.204) circle (2pt) node[right] {$v(1)$};

\draw[thick, domain=0:6] plot(\x,0.530*\x-0.353);
\draw[black,fill=black] (2,0.530*2-0.353) circle (2pt) node[] {};
\draw[thick, dashed] (2,0) node[below]{$p_I$}--(2,0.530*2-0.353) node[] {};
\draw[black,fill=black] (0,0.530*0-0.353) circle (2pt) node[left] {$v_{L}^I$};
\draw[black,fill=black] (6,0.530*6-0.353) circle (2pt) node[right] {$v_{H}^I$};

\draw[thick, domain=1.5:6] plot(\x,1.2*\x-3.204);
\draw[black,fill=black] (5,1.2*5-3.204) circle (2pt) node[] {};
\draw[thick, dashed] (5,0) node[below]{$p_N$}--(5,1.2*5-3.204) node[] {};
\draw[black,fill=black] (6,1.2*6-3.204) circle (2pt) node[right] {$v_{H}^N$};

\draw[thick, dashed] (6,0) node[below]{$1$}--(6,5.5) node[] {};

\end{tikzpicture}
\caption{the continuation value $v(p)$}
\label{fig:continuationvalue}
\end{figure}
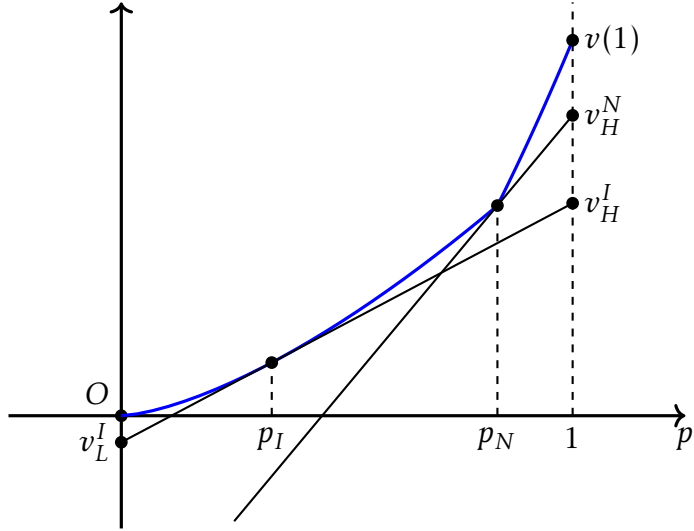

\begin{proofof}{Lemma \ref{lem:threshold}}
Denote by $p$ player $i$'s belief at the action history $a^{t-1}$. Let $s^*_I$ be the strategy in which player $i$ invests at $a^{t-1}$ and then never again. The conditional payoffs for this strategy are
\begin{align*}
    v^*_H \ge (1-\delta)(1-c -\varepsilon), \quad
    v^*_L = -(1-\delta)c.
\end{align*}
where the former holds true because the opponent invests with probability at least $1-\varepsilon$ in the state $\theta = H$.

Let $s_N$ be any strategy in which player $i$ does not invest at $a^{t-1}$. Then, because of the cool-off scheme, the payoffs for this strategy are
\begin{align}\label{eq:s-bound1}
    v^N_H \le \delta^{T_0}.
\end{align}
Let $s_I$ be any strategy in which player $i$ invests at $a^{t-1}$. Then
\begin{align}\label{eq:s-bound2}
    v^I_L \leq -(1-\delta)c.
\end{align}

Assume by contradiction that it is optimal to choose the strategy $s_I$ at belief $p_I$ and to choose the strategy $s_N$ at belief $p_N > p_I$. Recall that the continuation payoff $v(p)$, defined in the proof of Lemma \ref{lem:pure}, is supported by the line $p v^I_L+(1-p) v^I_H$ at the belief $p_I$ and by $p v^N_L+(1-p) v^N_H$ at $p_N$. Since it is non-decreasing and convex (Lemma \ref{lem:pure}), it follows that $v^I_H \le v^N_H$. Inequality \eqref{eq:s-bound1} implies that $v^I_H \le \delta^{T_0}$. By condition \eqref{eq:threshold}, we have that $v^I_H < v^*_H$. By inequality \eqref{eq:s-bound2} we have that $v^I_L \le v^*_L$, and thus the continuation strategy $s^*_I$ strictly dominates $s_I$ for any belief $p \neq 0$. Hence, the continuation strategy $s_I$ cannot be optimal in such a belief. Moreover, it cannot be optimal at the belief $p = 0$ either, since it is better to not invest. We have thus reached a contradiction.
\end{proofof}

\begin{proofof}{Corollary \ref{cor:monotone}}
It suffices to establish the corollary in the case where $s = t+1$. That is, $a^{s-1} = a^{t} = (a^{t-1}, (I,I))$. Now, we have, 
\begin{align*}
\mathbb{E}\bigl[p_{it+1} \mid \theta=H, a^{t-1}, a_{it}=I,a_{jt} = I\bigr] 
 &
\ge   \mathbb{E}\bigl[p_{it+1} \mid \theta=H, a^{t-1},a_{jt} = I\bigr]  
    \notag
\\  &
\ge   \mathbb{E}\bigl[p_{it+1} \mid \theta=H, a^{t-1}\bigr], 
    \notag
\end{align*}
where the first inequality follows from the fact that player $i$ uses a threshold rule at $a^{t-1}$ (Lemma \ref{lem:threshold}), hence conditioning on $a_{it}=I$ increases the expectation of players $i$'s private belief. The second inequality follows since player $j$ uses a threshold rule, hence conditional on  $a_{jt}= I$, player $i$ learns that player $j$'s belief is above a threshold and hence, his beliefs move upwards. Now, by the tower property of expectations, 
\begin{align*}
\mathbb{E}\bigl[p_{it+1} \mid \theta=H, a^{t-1}\bigr] =  \mathbb{E}\bigl[\mathbb{E}\bigl[p_{it+1} \mid \theta=H, a^{t-1}, p_{it}\bigr] \mid \theta = H, a^{t-1}\bigr]
\end{align*}
From Lemma \ref{lem:markovineq}, and since $p_{it}=\mathbb{P}[\theta=H \mid a^{t-1},x_i^t]$, it follows that $\mathbb{E}[p_{it+1} \mid \theta=H, a^{t-1}, p_{it}]\geq p_{it}$ almost surely, and so
\begin{align*}
\mathbb{E}\bigl[p_{it+1} \mid \theta=H, a^{t-1}\bigr] \geq \mathbb{E}\bigl[p_{it} \mid \theta = H, a^{t-1}\bigr].
\end{align*}
\end{proofof}

\begin{proofof}{Proposition \ref{prop:fixedpoint}}
  The strategy space $\Sigma_{i}$ is compact. By Claim
  \ref{clm:legal-closed-convex}, the set $\Sigma_{i}(T,\varepsilon)$ is
  closed, thus compact, and convex. Payoffs are continuous, which
  follows from the fact that they are discounted exponentially and
  that the topology on $\Sigma_i$ is the weak-* topology. Hence, by
  Glicksberg's Theorem, the game in which the players are restricted
  to $\cS(T,\varepsilon)$ has a Nash equilibrium
  $\sigma^{\ast} = (\sigma_{i}^{\ast},\sigma_{j}^{\ast}) \in
  \cS(T,\varepsilon)$.
  It follows from Propositions \ref{prop:legal-best-response} and
  \ref{prop:legal-best-response2} that this fixed point is, in fact, a
  Nash equilibrium of the unrestricted game, since any best response
  to any strategy in $\Sigma_{j}(T,\varepsilon)$ would be in
  $\Sigma_{i}(T,\varepsilon)$.
\end{proofof}

\begin{proofof}{Proposition \ref{prop:commonlearning}}
We prove the claim in the case of the state $\theta=H$.
The proof for the case of the state $\theta=L$ is analogous.
We proceed in two steps.

\paragraph{Step 1:}
Fix a strategy profile $\sigma$, and suppose that there exists an on-path action history $\bar{a}^{t-1} \in A^{t-1}$ such that
\begin{align*}
\mathbb{E}[p_{it} \mid \theta=H, \bar{a}^{t-1}] \ge 1-\varepsilon^{2}
\end{align*}
for each $i = 1,2$. From Lemma \ref{lem:markovineq}, it follows that
\begin{align}
\mathbb{P}[p_{it} > 1-\varepsilon \mid \theta=H, \bar{a}^{t-1}] \ge 1-\varepsilon \label{evidence}
\end{align} 
We will show that 
\begin{align}\label{eq:commonlearning}
	\mathbb{P}\bigl[C_{t}^{(1-\varepsilon)^2}(H) \mid \theta=H, \bar{a}^{t-1}\bigr]
\ge	(1-\varepsilon)^2.
\end{align}
Note that beliefs $p_{1t}$ and $p_{2t}$ are independent conditional on the state $\theta=H$ and action history $\bar{a}^{t-1}$, since private signals are conditionally independent.
Now, consider the event
\begin{align*}
	F_{t} = \bigl\{
		\omega \in \Omega : 
		a^{t-1}(\omega) = \bar{a}^{t-1},~\textup{$p_{it}(\omega) \ge 1-\varepsilon$ for each $i=1,2$}
	\bigr\}.
\end{align*}
Note that from inequality \eqref{evidence}, we have
$\mathbb{P}[F_{t} \mid \theta=H, \bar{a}^{t-1}] \ge
(1-\varepsilon)^{2}$.  Hence, it suffices to show that
$F_{t} \subseteq C_{t}^{(1-\varepsilon)^{2}}(H)$.  By the result of
\cite{monderersamet}, it suffices to show that the event $F_{t}$ is
$(1-\varepsilon)^{2}$-evident (i.e.,
$F_{t} \subseteq B_{t}^{(1-\varepsilon)^{2}}(F_{t})$) and that
$F_{t} \subseteq B_{t}^{(1-\varepsilon)^{2}}(H)$.
The latter follows since $F_{t} \subseteq B_{t}^{1-\varepsilon}(H) \subseteq B_{t}^{(1-\varepsilon)^{2}}(H)$.
Thus, we show the former.
At any $\omega \in F_{t}$, in period $t$ player $i$'s belief about the event $F_{t}$ is
\begin{align*}
	\mathbb{P}\bigl[F_{t} \mid \bar{a}^{t-1}, x_{i}^{t}\bigr]
\ge	p_{it}(\omega)
	\mathbb{P}\bigl[F_{t} \mid \theta=H, \bar{a}^{t-1}, x_{i}^{t}\bigr]
\ge	(1-\varepsilon)^{2}.
\end{align*}
This implies that $\omega \in B_{t}^{(1-\varepsilon)^2}(F_{t})$.
That is, $F_{t} \subseteq B_{t}^{(1-\varepsilon)^2}(F_{t})$.
Hence, inequality \eqref{eq:commonlearning} follows.
Moreover, since $C_{t}^{(1-\varepsilon)^2}(H) \subseteq C_{t}^{(1-\varepsilon)^{2}(1-\varepsilon^{2})}(H)$, we have that
\begin{align*}
	\mathbb{P}\bigl[C_{t}^{(1-\varepsilon)^{2}(1-\varepsilon^{2})}(H) \mid \theta=H, \bar{a}^{t-1}\bigr]
\ge	(1-\varepsilon)^{2}.
\end{align*}

\paragraph{Step 2:}
Fix any $\varepsilon \in (0,1)$, and let $\bar{T} \in \mathbb{N}$ be such that for each $i = 1,2$ and $t \ge \bar{T}$,
\begin{align*}
	\mathbb{E}\bigl[p_{it} \mid \theta=H\bigr]
\ge	1-\left(\frac{\varepsilon^{2}}{2}\right)^{2}.
\end{align*}
Note that there exists such a $\bar{T}$, since each player individually learns the state. By the tower property of conditional expectations, 
\begin{align*}
\mathbb{E}[p_{it} \mid \theta=H] = 	\mathbb{E}\bigl[\mathbb{E}\bigl[p_{it} \mid \theta=H, a^{t-1}\bigr] \mid \theta=H\bigr]
\end{align*}
From Lemma \ref{lem:markovineq}, we have that for each $i = 1,2$,
\begin{align*}
	\mathbb{P}\bigl[\mathbb{E}\bigl[p_{it} \mid \theta=H, a^{t-1}\bigr] \ge 1-\frac{\varepsilon^{2}}{2} \mid \theta = H\bigr]
\ge	1-\frac{\varepsilon^{2}}{2}.
\end{align*}
Then,
\begin{align*}
	\mathbb{P}\bigl[
	\textup{$\mathbb{E}\bigl[p_{it} \mid \theta=H, a^{t-1}\bigr] \ge 1-\varepsilon^{2}$ for each $i=1,2$} \mid \theta = H\bigr]
\ge 1-\varepsilon^{2}.
\end{align*}
Hence, for each $t \ge \bar{T}$,
\begin{align*}
	\mathbb{P}\bigl[C_{t}^{(1-\varepsilon)^{2}(1-\varepsilon^{2})}(H) \mid \theta = H\bigr]
\ge	(1-\varepsilon)^{2}(1-\varepsilon^{2}),
\end{align*}
which establishes the desired result. 
\end{proofof}

\section{Sequential equilibrium}

A \textit{behavioral strategy} for player $i$ is defined as a map from private histories to probabilities over actions i.e. $\sigma_i : \displaystyle\bigcup_{t \geq 1} X^t \times A^{t-1} \rightarrow \Delta(\{I,NI\})$. For a given strategy profile $\sigma = (\sigma_i,\sigma_j)$, we define \textit{continuation values} as $V^{\theta}_{i}(\sigma \mid h_t)$ for each $h_t = (a^{t-1},x_i^t,x_j^t) = (h_{it},h_{jt})$. This value represents the long-run expected payoff player $i$ would obtain at history $h_t$, if players play according to the strategies $\sigma$ and the state is $\theta$. At period $t$, depending on the private history, player $i$ has beliefs about the state of nature and private signals of his opponent $(\theta,x^t) \in \Theta \times X^t$. Player $i$'s \textit{belief function} is given by $\mu_{i} : \bigcup_{t} H_{it} \rightarrow \Delta(\Theta \times X^t) $. An \textit{assessment} is a pair $\langle \sigma,\mu \rangle$ consisting of a behavioral strategy profile $\sigma = (\sigma_i)_i$ and belief functions for each player $\mu = (\mu_i)_i$. We use sequential equilibrium as our solution concept, which we define in the current setting as follows: 
\begin{definition}\label{def:seqeqm}
An assessment $\langle \sigma^*,\mu^* \rangle$ is a \emph{sequential equilibrium} if: 
\begin{enumerate}
\item
\emph{(Sequential Rationality)}
The strategy profile $\sigma^*$ is sequentially rational for each player given beliefs $\mu^*$. This means, for each $h_{it} = (a^{t-1},x_i^t) \in H_{it}$ and $\sigma_i\in \Sigma_i$,
\begin{align}
    \sum_{(\theta,x_j^t)}
    \mu^*_{i}\bigl(h_{it}\bigr)\bigl(\theta,x_j^t\bigr)
    V^{\theta}_i\bigl(\sigma^* \mid a^{t-1},x_i^t,x_j^t\bigr)
\ge \sum_{(\theta,x_j^t)}
    \mu^*_{i}\bigl(h_{it}\bigr)\bigl(\theta,x_j^t\bigr)
    V^{\theta}_i\bigl(\sigma_i,\sigma^*_{-i} \mid a^{t-1},x_i^t,x_j^t\bigr).
\label{sequentialrationality}
\end{align}
\item
\emph{(Consistency)}
There exists a sequence of profiles of completely mixed strategies (that assigns positive probability to all actions at all private histories) $\{\sigma^n\}_{n \ge 1}$ such that the belief functions $\{\mu^n\}_{n\geq 1}$ uniquely induced by them satisfy  
\begin{align}
\lim_{n \rightarrow \infty}(\mu_{i}^{n}(h_{it}),\sigma_{i}^{n}(h_{it})) = (\mu^{*}_i(h_{it}),\sigma^{*}_i(h_{it}))  
\label{consistency}
\end{align} 
for all $h_{it} \in H_{it}$, and convergence is with respect to the product  topology on $\Delta(\Theta \times X^t) \times A_i$. 
\end{enumerate} 
\end{definition}

The definition of sequential equilibrium is borrowed from
\cite{kreps1982sequential} and is adapted to our setting. While their
definition corresponds to finite extensive-form games, the extension
to the current setting is justified since there are only countably
many information sets for each player and moreover, each information
set is finite.

The following Lemma reduces the problem of finding a sequential
equilibrium to that of finding a Nash equilibrium.
\begin{lemma}
  \label{lemma:sequential}
  For every mixed strategy Nash equilibrium $\sigma^*$ in our game there exists an
  outcome equivalent sequential equilibrium $(\sigma^{**},\mu^*)$.
\end{lemma}

\begin{proof}
  Suppose $\sigma^*$ is a mixed strategy Nash equilibrium. From Kuhn's theorem, let $\sigma_b^*$ be an outcome equivalent behavioral strategy profile. It follows that $\sigma_b^*$ is also Nash Equilibrium of the game. Let $\sigma^{**}$ equal $\sigma_b^*$ for on-path histories, and off-path,  let both
  players play $N$ under $\sigma^{**}$. $\sigma^{**}$ is still a Nash
  equilibrium (and clearly outcome equivalent) since now, off-path,
  players are guaranteed utility at most zero, whereas before
  utilities off-path were non-negative. Also, $\sigma^{**}$ is sequentially
  rational---for any belief system that is consistent with Bayes' rule
  on-path---since, if player $i$ never invests then player $j$'s
  unique best response is to also never invest.

  We construct the belief system $\mu^*$ as follows. Firstly, note
  that beliefs on-path are uniquely determined by Bayes' rule. Now,
  for any observable deviation made by the opponent, an action (one of
  $\{I,N\}$) was taken even though it is supposed to be chosen with
  probability zero. We assume that in $\mu^*$, upon observing such an
  action, players update their belief as if they observed the other
  action, which was supposed to be chosen with probability one.

  To see that $\mu^*$ satisfies consistency, for each $n > 1$, define
  the behavioral strategy $\sigma_i^n$ for each player $i = 1,2$ as follows
  \begin{equation}
    \sigma_i^n(a^{t-1},x_i^t) = \begin{cases} 1-\frac{1}{n} &\mbox{if } \sigma_i^{**}(a^{t-1},x_i^t)=1 \\ 
      \frac{1}{n} &\mbox{if } \sigma_i^{**}(a^{t-1},x_i^t)=0\\ 
      \sigma_i^{**}(a^{t-1},x_i^t) & o.w. \end{cases} 
  \end{equation}
  Note that $\sigma^n_i$ chooses each action at each private history
  with positive probability. In particular, $\sigma^n$ differs from
  $\sigma^{**}$ in that whenever under $\sigma^{**}$ player $i$
  chooses an action with probability one, under $\sigma^n$ player $i$
  chooses that same action instead with probability
  $1-\frac{1}{n}$. Note that this change does not depend on the
  agent's private history. Hence, when an observable deviation occurs
  under $\sigma^{**}$, the beliefs of the players are unchanged by
  assumption and are close to the beliefs generated by $\sigma^n$ as
  $n$ tends to infinity.
\end{proof}

\bibliographystyle{plainnat}
\bibliography{repeated.bib}

\end{document}